\documentclass[10pt,a4paper]{article}

\usepackage{amsmath,amssymb,amsthm,mathrsfs}
\usepackage{stmaryrd}  
\usepackage{graphicx}
\usepackage{xcolor}
\usepackage{soul}
\usepackage{fullpage}

\newlength{\defbaselineskip}
\newcommand{\setlinespacing}[1]%
           {\setlength{\baselineskip}{#1 \defbaselineskip}}

\theoremstyle{plain}
\newtheorem{theorem}{Theorem}[section]
\newtheorem{lemma}[theorem]{Lemma}
\newtheorem{proposition}[theorem]{Proposition}
\newtheorem{corollary}[theorem]{Corollary}
\newtheorem{definition}[theorem]{Definition}

\theoremstyle{definition}
\newtheorem{assumption}{Standing Assumption}
\newtheorem{ass}[theorem]{Assumption}

\theoremstyle{remark}
\newtheorem{remark}[theorem]{Remark}

\numberwithin{equation}{section}

\DeclareMathOperator*{\esssup}{ess\,sup}

\newcommand{\cF}{\mathcal{F}}

\newcommand{\bE}{\mathbb{E}}

\newcommand{\bP}{\mathbb{P}}
\newcommand{\bR}{\mathbb{R}}
\newcommand{\bN}{\mathbb{N}}

 \allowdisplaybreaks
\begin{document}
\title{Mean Field Portfolio Games with Epstein-Zin Preferences
}

\author{Guanxing Fu\footnote{Department of Applied Mathematics, and Research Centre for Quantitative Finance, The Hong Kong Polytechnic University,
         Hung Hom, Kowloon, Hong Kong; email: guanxing.fu@polyu.edu.hk. Fu gratefully acknowledges financial support through NSFC Grant No. 12471453, Hong Kong RGC (GRF) Grant No. 15218825, and internal grants from The Hong Kong Polytechnic University and Shenzhen Research Institute.}  ~ and ~ Ulrich Horst\footnote{Department of Mathematics, and School of Business and Economics, Humboldt-Universit\"at zu Berlin,
         Unter den Linden 6, 10099 Berlin, Germany; email: horst@math.hu-berlin.de. Horst gratefully acknowledges financial support by the Deutsche Forschungsgemeinschaft through CRC TRR 190.}
}

\maketitle

\begin{abstract}
	
	We study mean field portfolio games under Epstein-Zin preferences, which naturally encompass the classical time-additive power utility as a special case. 
 In a general non-Markovian framework, we establish a uniqueness result by proving a one-to-one correspondence between Nash equilibria and the solutions to a class of BSDEs.
	A key ingredient in our approach is a necessary local stochastic maximum principle, applied to log-wealth, tailored to Epstein-Zin utility, and a nonlinear transformation. In the deterministic setting, we further derive an explicit closed-form solution for equilibrium investment and consumption policies. The strength of our approach is further illustrated by two special cases: (i) in the power utility setting without consumption, we obtain the same one-to-one correspondence as in Fu and Zhou \cite{Fu-2021} under exactly the same assumption, but without invoking the dynamic programming principle in Espinosa and Touzi \cite{ET-2015}; and (ii) in the power utility setting with both investment and consumption, we strengthen the correspondence result of Fu \cite{Fu2023}, by proving a genuine one-to-one relation in the BMO space, where both the equilibrium strategy and the associated BSDE components belong to BMO.

\end{abstract}



\section{Introduction}

In their seminal paper  \cite{EZ-1989}, Epstein and Zin developed a class of recursive preferences over intertemporal consumption lotteries that permit risk attitudes to be disentangled from the degree of intertemporal substitutability. Utility optimization problems with Epstein-Zin preferences have since been analyzed in various settings by many authors, including \cite{Aurand-Huang-2023,DE-1992,HHJ-2023,HHJ-2023a,HHJ-2025,HLT-2024,Kraft-2017,Kraft-2010,Kraft-2013,MX2018,Melnyk2020,SS98,Seifried-2015,Xing2017}. 

In this paper, we consider a class of mean-field portfolio games under Epstein-Zin utility with relative performance concerns in a general stochastic framework. Mean field games (MFGs) are a powerful tool to analyze strategic interactions in large populations when each individual player has only a small impact on the behavior of other players. Introduced independently by Huang, Malham\'e and Caines \cite{HMC-2006} and Lasry and Lions \cite{LL-2007}, MFGs have been successfully applied to many economic and engineering problems. To name only a few of them among many others, applications include optimal trading under market impact \cite{FGHP-2018, FHX-2022,FHH-2023} to risk management \cite{CFMS-2016}, and from principal agent problems \cite{EMP-2016}, to optimal exploitation of exhaustible resources \cite{CS-2015, PUR}. These examples are illustrative rather than exhaustive, as a comprehensive survey lies beyond the scope of this paper.

Mean-field portfolio games, including finite-player games, with relative performance criteria where all players trade a common stock were first analyzed by Espinosa and Touzi \cite{ET-2015}. In a complete market setting, they established the existence of a unique Nash equilibrium for general time additive utility functions; in an incomplete market settings with player-specific portfolio constraints, they proved the uniqueness under exponential utility. Frei and dos Reis \cite{FR-2011} addressed the existence of equilibria in similar games. By leveraging the dynamic programming principle (DPP) from \cite{ET-2015}, they established a one-to-one correspondence between Nash equilibria and solutions to a multidimensional backward stochastic differential equation (BSDE) of quadratic growth and constructed a counterexample to show that equilibria in games with performance concerns may not exist. Lacker and Zariphopoulou \cite{LZ-2019} later considered a model where different players trade different stocks whose price dynamics is correlated through a common noise process. They identified a unique constant equilibrium in a time-homogeneous setting. Following up on these works, portfolio games with {\sl time-additive utilities} have been extended by many authors to an array of different settings and features, including jump signal, price impact, external habit formation, consumption, non-Markovian framework, Itô-diffusion environment, partial information, non-exponential discounting, interaction through a graphon; see \cite{Bank2026,Baeuerle2024,Bo2024,RP-2021b,Fu2023,Fu-2021,HZ-2021,Huang2026,LS-2020,Liang2024,Tangpi2024}. An exception is dos Reis and Platonov \cite{RP-2021}, who employ forward utility. All these works consider specific utility and linear coupling (multiplicative coupling should also be considered as linear). For results with general utility and/or nonlinear coupling, one may refer to \cite{Zaripopoulou2026,Souganidis2024}, where the game with general utility and coupling is introduced directly through master equations, and \cite{Zariphopoulou2024}, where a notion of forward MFG is proposed. 


To the best of our knowledge, portfolio games with Epstein-Zin preferences and performance concerns have only been studied in \cite{Riedel-2024, Wang2023}. The work \cite{RP-2021b} shows that the forward performance framework, both with and without competition, allows agents to exhibit behavior similar to that under Epstein-Zin utility by disentangling risk aversion from the elasticity of intertemporal substitution. However, it does not provide a systematic study of portfolio games under Epstein-Zin preferences.

We establish an existence of equilibrium result for such games with deterministic, though possibly time-varying model parameters, and a uniqueness of equilibrium result for general stochastic settings. Specifically, we prove that Nash equilibria in mean-field portfolio games with Epstein-Zin preferences are uniquely characterized in terms of a solution to a certain quadratic BSDE. Similar characterization results are often implicitly assumed - though rarely proved - in the literature.  For instance, Dianetti, Riedel and Stanca \cite{Riedel-2024} prove the existence of a unique simple equilibrium assuming that optimizing the Epstein-Zin utility index is equivalent to optimizing the driver of some BSDE. This is not always the case, though. Our uniqueness result shows that the simple equilibrium obtained in \cite{Riedel-2024} is indeed the unique bounded equilibrium.  


 
The fact that any solution to a certain BSDE yields a Nash equilibrium is established using the martingale optimality principle (MOP). This approach is consistent with the methodology adopted in prior works such as \cite{Fu2023,Fu-2021}, and extends the results in \cite{HIM-2005} beyond the benchmark case of time-additive utilities and the results in  \cite{Xing2017} to models with mean-field interaction. 
The sufficient condition yields an existence result, as most results in literature do. To establish the uniqueness of equilibria, the crucial step is to complement the sufficient condition by a characterization of any best response or equilibrium strategy in terms of the same BSDE leading to the sufficient condition. This one-to-one correspondence translate the analysis for unique equilibrium to that for a unique solution to a certain BSDE.

To prove the necessary characterization, we apply a {\it local} stochastic maximum principle to the log-wealth process. 
The key observation is that stochastic control problems involving Epstein-Zin preferences naturally lead to an FBSDE system as their state dynamics. Crucially, the driver of the BSDE component characterizing Epstein-Zin utility is not Lipschitz (see e.g. \cite{ElKaroui-2001} for the maximum principle for recursive utilities with Lipschitz drivers) which prevents us from using established  stochastic maximum principles for FBSDEs. 



To connect the adjoint processes arising from the stochastic maximum principle with the desired BSDE characterization, we introduce a nonlinear transformation that maps the adjoint processes to a candidate process $Y$, which we expect to solve the targeted BSDE. Subsequently we demonstrate that the difference between the Epstein-Zin utility associated with the optimal investment–consumption strategy and a certain process involving the optimal wealth and the transformed process $Y$ satisfies a linear BSDE with zero terminal condition. This crucial step confirms that the candidate process $Y$ coincides with the solution to the BSDE arising from the MOP, thereby establishing a one-to-one relation between the Nash equilibrium and the solution to certain BSDE. 


Our paper relates to two strands of the literature, and the comparison highlights our twofold contribution.
First, our method for establishing the one-to-one correspondence is novel. In the earlier works of the first author \cite{Fu2023, Fu-2021}, the necessary characterization relies on the DPP developed in \cite{ET-2015}. For Epstein–Zin utility, however, such a DPP is not available. Instead, we adopt a maximum-principle-based approach, which is significantly easier to verify than the DPP. Moreover, our method is not only more tractable and technically simpler, but also yields stronger results. In particular, we recover the one-to-one correspondence result in \cite{Fu-2021} for the power utility case with investment only, under exactly the same assumptions. Furthermore, we improve the correspondence result in \cite{Fu2023} by establishing a genuine one-to-one correspondence in the random setting.
Second, although uniqueness, especially the one-to-one correspondence, is often implicitly assumed in the literature, it is in fact highly nontrivial. This is particularly true in our setting, since even when it reduces to the time-additive case our model does not satisfy the classical Lasry–Lions monotonicity condition or the displacement monotonicity condition, which are commonly used to guarantee uniqueness in MFGs, nor does it satisfy the supermodularity condition generally, which typically yields the existence of minimal and maximal equilibria as an alternative to uniqueness; see e.g. \cite{Cardaliaguet2019,Dianetti2023,Gangbo2022}. As a corollary, our result confirms the uniqueness claim for the simple equilibrium considered in \cite{Riedel-2024}.

The remainder of the paper is organized as follows. Section 2 recalls the single player benchmark model with Epstein-Zin utility along with its game-theoretic extensions. Section \ref{sec:one-to-one} establishes a one-to-one correspondence between the Nash equilibrium and the solution to a certain BSDE. The sufficient condition—showing that a solution to the BSDE yields a Nash equilibrium—is presented in Section~\ref{sec:sufficient}. The necessary condition—showing that any Nash equilibrium must correspond to a BSDE solution—is developed in Section~\ref{sec:necessary}. 
The uniqueness of the Nash equilibrium result is established in Section~\ref{sec:uniqueness}. In Section~\ref{sec:existence}, we provide a closed-form expression for the equilibrium investment and consumption strategy in a  deterministic setting. Last but not least, in Section \ref{sec:power} we recover the necessary characterization of the best response in the power utility setting and further demonstrate the power of our approach.



 \section{The mean field game with Epstein-Zin utility} \label{sec:EZ-utility}

In this section we recall the benchmark model of a single player optimization problem with Epstein-Zin preferences along with a game-theoretic extension with performance concerns. 

We fix a common time horizon $[0,T]$ for all investors and assume that all random variables are defined on a filtered probability space $(\Omega, \cF, (\cF_t)_{t \in [0,T]}, \bP)$ that carries independent Brownian motions $W^0, W^1,\cdots, W^N$ for some $N \in \bN$. The Brownian motions $W^1,\cdots, W^N$ capture {\sl idiosyncratic noises} in an $N$-player model while the Brownian motion $W^0$ captures a {\sl common noise} shared by all agents.  

 \subsection{Single agent benchmark model}

In a single agent benchmark model the agent can invest in a risky asset whose price process $(S_t)$ follows the SDE
 \[
 	\frac{dS_t}{S_t} = h_t dt + \sigma_t dW_t + \sigma^0_t dW^0_t, \qquad t \in [0,T].
 \] 
 
 \begin{assumption}
	The processes $h$, $\sigma$ and $\sigma^0$ are progressively measurable and bounded, and $\Sigma^2:=(\sigma)^2+(\sigma^0)^2$ is strictly positive.
 \end{assumption}
We denote by $\pi$ and $c$ a pair of progressively measurable processes that represent the investor’s investment strategy and consumption strategy, respectively. Specifically, $\pi_t$ denotes the proportion of wealth invested in the risky stock and $c_t$ denotes the proportion of wealth consumed at time $t \in [0,T]$. The corresponding wealth process $X^{\pi, c}$ then evolves according to the SDE
\begin{equation*}
	\left\{\begin{split}
	dX^{\pi, c}_t =&~   \pi_t X_t^{\pi,c} h_t\, dt + \pi_t X_t^{\pi,c} \sigma_t\, dW_t + \pi_t X_t^{\pi,c} \sigma^0_t\, dW^0_t - c_t X_t^{\pi,c}\, dt, \qquad t \in (0,T], \\
	X^{\pi,c}_0= &~x>0.
	\end{split}\right.
\end{equation*}
In a model without mean-field interaction, the dynamics of the investor's Epstein-Zin utility from dollar consumption $cX^{\pi,c}$ is given by  
 \begin{equation}\label{EZ-0}
 	\begin{split}
 		\tilde V^{\pi,c}_t=&~\mathbb E\left[ \left.  \int_t^T   f(c_sX^{\pi,c},\tilde V^{\pi,c}_s)\,ds + \alpha U(c_TX_T^{\pi,c}) \right|\mathcal F_t     \right], \qquad t \in [0,T],
 	\end{split}
 \end{equation}
where the {\sl aggregator function} $f:[0,\infty) \times (-\infty,0] \to \bR$ is given by
\begin{equation}\label{aggregator}
	f(C,V) = \frac{	 \delta C^{1-\frac{1}{\psi}}		}{1-\frac{1}{\psi}} \{  (1-\gamma)V     \}^{1-\frac{1}{\vartheta}}  - \delta\vartheta V,  \quad \vartheta := \frac{1 - \gamma}{1 - \frac{1}{\psi}}.
\end{equation}
Here $\delta>0$ represents the {discount rate}, $\gamma >0$ specifies the {relative risk aversion}, $\psi >0$ represents the {elasticity of intertemporal substitution} (EIS), $\alpha > 0$ is a {rate of bequest}, and  the {\sl bequest function} $U$ is of power type, i.e.
\begin{equation}\label{power}
	U(x)=\frac{x^{1-\gamma}}{1-\gamma}. 
\end{equation}
\begin{assumption}
We assume throughout this paper that\footnote{The choice in \eqref{parameters} allows for a unified treatment of Epstein-Zin and power utility. Technically, our results also hold under the conditions $\psi \gamma \leq 1$ and $\psi \leq 1$. 
}
\begin{equation}\label{parameters}
	{\psi \gamma \geq 1, \quad \psi \geq 1. }
\end{equation}
\end{assumption}
It is typically assumed that all wealth is consumed at the terminal time, i.e. the terminal consumption rate $c_T=1$. 



 \subsection{Game theoretic extensions}

The single agent model with Epstein-Zin preferences has recently been extended to multi-player models and MFGs with relative performance concerns by \cite{Riedel-2024}. Following their approach
, we consider an $N$-player model where all players share identical Epstein-Zin preferences, and where each player $i=1,\cdots, N$ can trade a risky stock whose dynamics is given by 
 \[
 	\frac{dS^i_t}{S^i_t} = h^i_t dt + \sigma^i_t dW^i_t + \sigma^{i0}_t dW^0_t, \qquad t \in [0,T].
 \] 
The wealth dynamics is defined as in the single player case. Given the profile $c=(c^1,\cdots,c^N)$ and $\pi=(\pi^1,\cdots,\pi^N)$, and the associated wealth profile $(X^{1},\cdots,X^{N})$, we denote by 
\[
	\bar C^i_t :=  \left( \prod_{j \neq i} C^j_t \right)^{\frac{1}{N-1}}, \quad \bar X^i_t :=  \left( \prod_{j \neq i} X^j_t \right)^{\frac{1}{N-1}}, \qquad t \in [0,T]
\]
the ergodic average consumption and wealth, respectively, of the players'$i$'s competitors, where $C^i=c^i X^i$. Consider utility functionals of the form
 \begin{equation}\label{EZ-2}
 	\begin{split}
 		V^{i,\pi,c}_t=&~\mathbb E\left[ \left.  \int_t^T  f_i(C^i_s (\bar C^i_s)^{-\theta_i},V^{i,\pi,c}_s)\,ds + \alpha_i U_i(X^{i}_T \bar X^i_T ) \right|\mathcal F_t     \right], \qquad t \in [0,T],
 	\end{split}
 \end{equation}
 where $f_i$ and $U_i$ are defined as in \eqref{aggregator} and \eqref{power}, with $\delta$, $\psi$ and $\gamma$ replaced by $\delta_i$, $\psi_i$ and $\gamma_i$, respectively.
 \begin{remark}
We refer to \cite{Riedel-2024} and references therein for a detailed economic motivation of the above preference functional. We emphasize, though, that our approach differs from \cite{Riedel-2024} in one important aspect. In \cite{Riedel-2024} the authors benchmark own consumption and wealth of each player against the ergodic averages $\left( \prod_{j=1}^N C^j \right)^{\frac{1}{N}}$ and $\left( \prod_{j=1}^N X^j \right)^{\frac{1}{N}}$. In particular, own consumption and wealth is included in the benchmark.  We believe that our approach of only considering average consumption and wealth of competitors captures better the idea of relative performance concerns. For \emph{time-additive utilities}, the two approaches result in equivalent optimization problems, as the problems can be transformed into each other  by adjusting model parameters such as the risk aversion coefficient and the competition parameter; see \cite[Remark 2.5]{LZ-2019} for details. For Epstein-Zin utilities, such parameter distortions lack economic intuition/justification and the two approaches may be no longer equivalent (except in the limit when the number of players tends to infinity).
 \end{remark}	

To study the MFG version of the $N$-player model, we fix an $\cF^0$-progressively measurable stochastic process  $\nu=(\nu_t)$ (``mean field externality'')  that captures the impact of aggregate intertemporal and terminal consumption on a representative player's utility, where $\mathcal F_0$ is the filtration generated by the common noise $W^0$. Specifically, in the MFG the representative investor's utility from consumption is then given by 
 \begin{equation}\label{EZ-I}
 	\begin{split}
 		V^{\pi,c}_t=&~\mathbb E\left[ \left.  \int_t^T  f(c_sX_s^{\pi,c}\nu^{-\theta}_s,V^{\pi,c}_s)\,ds + \alpha U(X^{\pi,c}_T \nu^{-\theta}_T ) \right|\mathcal F_t     \right], \qquad t \in [0,T].
 	\end{split}
 \end{equation}
 where $X^{\pi,c}$ is the wealth process associated with $(\pi,c)$,
 and with a slight abuse of notation, we continue to denote by $\mathcal F$ the filtration generated by the idiosyncratic noise $W$ and the common noise $W^0$ in the MFG setting, as no confusion will arise from the context.

Before introducing the space of admissible strategies, we first introduce several spaces of progressively measurable processes:
\begin{align*}
        &~L^\infty=\left\{ y:      \esssup_{\omega,t}|y_t|<\infty            \right\},\quad L^\infty_{\text{log},+}=\left\{ y: y\text{ is positive-valued, and }\log y\in L^\infty \right\}, \\
        &~H^2_{BMO}=\left\{ y: \esssup_{\mathrm{ stopping ~time~} \tau}\mathbb E\left[\left. \int_\tau^T y^2_s\,ds\right|\mathcal F_\tau   \right]<\infty \right\},\\
        &~ L^{\beta,T}_+=\left\{   y: y\text{ is positive-valued, and } \mathbb E\left[ \int_0^T y^\beta_t\,dt + y_T^\beta   \right]<\infty      \right\},\\
        &~S^\beta=\left\{ y: \mathbb E\left[ \sup_{0\leq t\leq T}|y_t|^\beta \right]<\infty \right\},\quad M^\beta=\left\{    y: \mathbb E\left[ \left(\int_0^T |y_t|^2\,dt\right)^\beta   \right]<\infty \right\}.
\end{align*} 
To simplify notation, we do not explicitly indicate the underlying filtration in the space notation, since it is clear that all processes in the MFG setting are $\mathcal F$-adapted, with the exception of $\nu$, which is $\mathcal F^0$-adapted.

\begin{definition}\label{def:admissibility}
An investment-consumption rate pair $(\pi,c)$, is called {\sl admissible} if it satisfies the following properties: 
 \begin{enumerate}
	\item[(i)] 
    All terminal wealth is consumed, i.e. $c_T=1$, equivalently, $X^{\pi,c}_T=C_T$;
 	
    \item[(ii)] The consumption rate process $c$ belongs to $L^\infty_{\mathrm{log},+}$;
 	\item[(iii)] The investment rate process $\pi$ belongs to $H^2_{BMO}$; 
    \item[(iv)] The dollar consumption process $C\in \bigcap_{\beta\in\mathbb R}L^{\beta,T}_+$. 
\end{enumerate} 
\end{definition}
\begin{remark}
The conditions $\pi \in H^2_{BMO}$ and $c^* \in L^\infty_{\text{log},+}$ are standard in the literature on additive utility; see, for instance, \cite{Fu-2021, Fu2023}. The integrability condition $C \in \bigcap_{\beta\in\mathbb R}L^{\beta,T}_+$ arises from the literature on Epstein-Zin utility (e.g., \cite{Seifried-2015}) and is specifically tailored to that framework. When Epstein-Zin utility reduces to power utility, however, this condition on $\bigcap_{\beta>0}L^{\beta,T}_+$ is no longer required. We refer to Section \ref{sec:power} for further discussion.\end{remark}
The set of admissible investment-consumption pairs is denoted ${\cal A}$. Fixing $\nu$, the representative investor's optimization problem is then given by
\begin{equation}\label{utility-max}
		\sup_{(\pi,c) \in {\cal A}}V^{\pi,c}_0 =\sup_{(\pi,c) \in {\cal A}}\mathbb E\left[  \int_0^T f(c_sX_s^{\pi,c}\nu^{-\theta}_s,V^{\pi,c}_s)  \,ds + \alpha U(X^{\pi,c}_T\nu^{-\theta}_T)  \right]. 
\end{equation}
Let $\big( c^*(\nu), \pi^*(\nu) \big)$ be an optimal consumption-investment pair, given $\nu$, with corresponding wealth process $X^*(\nu)$. In a mean-field equilibrium the expected optimal consumption and wealth processes coincides the ``anticipated'' externalities. 

\begin{definition}\label{def:equilibrium}
An $\cF^0$-progressively measurable stochastic process $\nu^*=(\nu^*_t)$ forms a mean-field equilibrium if a.s.
\[
	\widehat \nu^*_t = \bE[\widehat {c^*_tX^*_t}(\nu^*)| \cF^0_t],~0\leq t<T, \quad \mbox{and} \quad \widehat\nu^*_T = \bE[\widehat X^*_T(\nu^*)| \cF^0_T]
\]  
where $\widehat \xi := \log \xi$ denotes the logarithm of a strictly positive random variable $\xi$.
\end{definition}
In what follows we prove that at most one (in a certain class) mean-field equilibrium can exist under our standing assumptions. For the special case of deterministic model parameters we subsequently prove the existence (and hence uniqueness) of an equilibrium. Having solved the MFG, solving finite player game requires only minor adjustment of previously given arguments.

 
 \section{Characterization of mean-field equilibria}\label{sec:one-to-one}
 In this section we establish our characterization of equilibrium result. We prove that any mean-field equilibrium satisfies a certain BSDE and that, conversely, any solution to the said BSDE yields a mean-field equilibrium. 
 \subsection{From BSDE to Nash equilibrium}\label{sec:sufficient}
 
To identify the FBSDE system and subsequently the BSDE that any mean-field equilibrium must satisfy we first solve the representative investor's optimization problem for any fixed externality using the MOP. Subsequently, we solve the fixed-point problem to characterize MFG equilibria. 
 
 \subsubsection{Optimization}\label{sec:MOP}
 
 To solve the representative investor's utility optimization problem for a given externality $\nu\in \bigcap_{\beta\in\mathbb R}L^{\beta,T}_+$, we introduce the auxiliary process
 \begin{equation}\label{eqn:R}
 		R^{\pi,c}_t=\alpha \frac{ (X^{\pi,c}_t)^{1-\gamma}  }{1-\gamma} e^{Y_t} + \int_0^t f\left(c_sX_s\nu_s^{-\theta},  \alpha \frac{(X^{\pi,c}_s)^{1-\gamma}}{1-\gamma}e^{Y_s}\right)\,ds, \qquad t \in [0,T],
 \end{equation}
 where $Y$ denotes the first component of the solution to a BSDE of the form 
 \begin{equation} \label{BSDE1}
 		-dY_t=g_t\,dt-Z_t\,dW_t-Z^0_t\,dW^0_t,\qquad Y_T=-\theta(1-\gamma)\log\nu_T,
 \end{equation}
where the driver $g$ is to be determined such that the following conditions are satisfied:  
\begin{itemize}
 	\item for each admissible investmet-consumption pair $(\pi,c)$ the process $R^{\pi,c}$ is a supermartingale,
 	\item for some admissible investmet-consumption pair $(\pi^*,c^*)$ the process $R^{\pi^*,c^*}$ is a martingale,
 	\item the initial values $R_0^{\pi,c}$ are independent of the choice of $(\pi,c)$. 
 \end{itemize}
 The following remark justifies our approach. It shows that under the above assumptions, the pair $(\pi^*,c^*)$ is an optimal strategy.
 \begin{remark}
 	From the definition of the utility processes $V^{\pi,c}$ and $R^{\pi,c}$, we know that
 	\begin{itemize}
 		\item the process $V^{\pi,c}_t+\int_0^t f(c_sX_s^{\pi,c}\nu_s^{-\theta},V_s^{\pi,c})\,ds$ is a local martingale
 		\item the process $\alpha\frac{(X^{\pi,c}_t)^{1-\gamma}}{1-\gamma} e^{Y_t} + \int_0^t f(c_sX_s^{\pi,c}\nu_s^{-\theta}, \alpha \frac{(X^{\pi,c}_s)^{1-\gamma}}{1-\gamma}e^{Y_s}  )\,ds$ is 
		a local supermartingale,
	\end{itemize} 
and that
\[
	V^{\pi,c}_T = \alpha\frac{(X^{\pi,c}_T)^{1-\gamma}}{1-\gamma} e^{Y_T}. 
\]	

Thus, by the comparison principle in \cite[proof of Proposition 2.2, Step 3]{Xing2017}, it holds that
\[
		\alpha\frac{(X^{\pi,c}_t)^{1-\gamma}}{1-\gamma} e^{Y_t} \geq V^{\pi,c}_t, \qquad t \in [0,T].
\]
Since the left-hand side of the above inequality is independent of $(\pi,c)$ at time $t=0$ we see that
\[		
	\alpha	\frac{x^{1-\gamma}}{1-\gamma}e^{Y_0} \geq \sup_{\pi,c} V_0^{\pi,c}.
\]
Since $R^{\pi^*,c^*}$ is assumed to be a martingale, the following equality holds for some process $(\zeta,\zeta^0)$:
\[
	\alpha \frac{ (X_t^{\pi^*,c^*})^{1-\gamma}  }{1-\gamma} e^{Y_t} + \int_0^t f\left(c^*_sX^*_s \nu_s^{-\theta},  \alpha \frac{(X_s^{\pi^*,c^*})^{1-\gamma}}{1-\gamma}e^{Y_s}\right)\,ds=\int_0^t\zeta_s\,dW_s+\int_0^t \zeta^0_s\,dW^0_s, \quad t \in [0,T].
\]

It follows that the process 
\[
	\alpha \frac{ (X^{\pi^*,c^*})^{1-\gamma}  }{1-\gamma} e^{Y} 
\]	
satisfies the recursive utility equation. As a result, by \cite[Theorem 3.1]{Seifried-2015} it holds that
\[
	V_t^{\pi^*,c^*} = \alpha \frac{ (X_t^{\pi^*,c^*})^{1-\gamma}  }{1-\gamma} e^{Y_t}, \qquad t \in [0,T].
\]	
In particular, 
$
	V_0^{\pi^*,c^*}=\alpha \frac{ x^{1-\gamma}  }{1-\gamma} e^{Y_0},
$	
from which we conclude that $(\pi^*,c^*)$ is an optimal strategy. 
 \end{remark}
 
 It remains to determine the driver $g$ of the BSDE \eqref{BSDE1}.  By It\^o's formula, 
 \begin{equation*}
 	\begin{split}
 		 & d \Big( \frac{ (X^{\pi,c}_t)^{1 -\gamma }}{1-\gamma} e^{Y_t} \Big) \\ 
		 =&~ \frac{ (X^{\pi,c}_t)^{1-\gamma} }{1-\gamma} e^{Y_t}\left(	-g_t+\frac{Z^2_t+(Z^0_t)^2}{2}	\right)\,dt +\frac{(X^{\pi,c}_t)^{1-\gamma}}{1-\gamma} e^{Y_t}Z_t\,dW_t+\frac{ (X^{\pi,c}_t)^{1-\gamma} }{1-\gamma}e^{Y_t} Z^0_t\,dW^0_t\\
 		&~+e^{Y_t} \left\{   	\pi_t h_t(X^{\pi,c}_t)^{1-\gamma} - c_t(X^{\pi,c}_t)^{1-\gamma}-\frac{\gamma}{2}(X^{\pi,c}_t)^{1-\gamma}\pi^2_t( \sigma^2_t+(\sigma^0_t)^2			)	  \right\}\,dt \\
		& ~ + e^{Y_t} (X^{\pi,c}_t)^{1-\gamma} \pi_t\left(  \sigma_t\,dW_t+\sigma^0_t\,dW^0_t   \right) +e^{Y_t} Z_t (X^{\pi,c}_t)^{1-\gamma}\pi_t\sigma_t\,dt + e^{Y_t} Z^0_t (X^{\pi,c}_t)^{1-\gamma}\pi_t\sigma^0_t\,dt,
 	\end{split}
 \end{equation*}
from which we see that  
 \begin{equation*}
 	\begin{split}
 		&~d\big( R^{\pi,c}_t \big)\\
 		=&~\alpha \frac{(X^{\pi,c}_t)^{1-\gamma}}{1-\gamma} e^{Y_t}\left\{  -g_t+\frac{Z^2_t+(Z^0_t)^2}{2}	\right. \\
		& \qquad \left. +(1-\gamma)\Big(  \pi_t h_t - c_t -\frac{\gamma}{2}\pi^2_t(\sigma^2_t+(\sigma^0_t)^2) 	+ Z_t\pi_t\sigma_t+ Z^0_t\pi_t\sigma^0_t		    \Big)	         \right\}\,dt \\
 		&~+\alpha\left(\frac{ (X^{\pi,c}_t)^{1-\gamma} }{1-\gamma}e^{Y_t}  	Z_t +  e^{Y_t}(X^{\pi,c}_t)^{1-\gamma}\pi_t\sigma_t	\right)	     \,dW_t+ \alpha\left(  \frac{ (X^{\pi,c}_t)^{1-\gamma} }{1-\gamma} e^{Y_t}   Z^0_t+ e^{Y_t}(X^{\pi,c}_t)^{1-\gamma} \pi_t\sigma^0_t   \right)\,dW^0_t\\
 		&~+\left\{    \frac{\delta(c_tX_t^{\pi,c} \nu_t^{-\theta})^{1-\frac{1}{\psi}}}{	1-\frac{1}{\psi}			}		\left\{   \alpha  (X^{\pi,c}_t)^{1-\gamma} 	e^{Y_t}	   \right\}^{1-\frac{1}{\vartheta}}		-\delta\vartheta\alpha \frac{ (X^{\pi,c}_t)^{1-\gamma} }{1-\gamma}		e^{Y_t}		          \right\}\,dt.
 	\end{split}
\end{equation*}
The driver of the above BSDE is strictly concave in the investment-consumption pair $(\pi_t,c_t)$. Setting its partial derivatives w.r.t.~$\pi_t$ and $c_t$ equal to zero, the pointwise maximizer on $[0,T)$ is given by
 \begin{equation}\label{consumption-optimization}
 	\begin{split}
 		\pi^* &~ =\frac{ h+\sigma Z+\sigma^0 Z^0}{\gamma(	\sigma^2+(\sigma^0)^2	)},  \\
 		c^* & ~ = \delta^\psi (\nu^{-\theta})^{\psi-1}(\alpha e^Y)^{-\frac{\psi}{\vartheta}}. 
 	\end{split}
 \end{equation}
 Taking these quantities back into the driver of $R^{\pi,c}$ and letting  
 \begin{equation}\label{driver-MOP}
 	\begin{split}
 		g:=&~\frac{Z^2+(Z^0)^2}{2}  +\frac{1-\gamma}{2\gamma} \frac{	(h+\sigma Z+\sigma^0Z^0)^2		}{\sigma^2+(\sigma^0)^2} +\frac{1-\gamma}{\psi-1} \delta^\psi(\nu^{-\theta})^{\psi-1} (\alpha e^Y)^{ -\frac{\psi}{\vartheta} } -\delta\vartheta,
 	\end{split}
 \end{equation}
we see that the driver of the BSDE \eqref{eqn:R} vanishes for the pair $(\pi^*,c^*)$ and is non-positive for any admissible pair $(\pi,c)$. Thus, the process $R^{\pi^*,c^*}$ is a martingale while the process $R^{\pi,c}$ is a supermartingale for any admissible strategy.  

\begin{remark}\label{remark:integrability-nu}
In equilibrium, we require $\nu^*_t=\exp\left(\mathbb E[\log C^*_t|\mathcal F^0_t]    \right):=\exp\left(\mathbb E[\log c^*_tX^*_t|\mathcal F^0_t]    \right)$ for all $t\in[0,T]$; refer to Definition \ref{def:equilibrium}. By Jensen's inequality, it holds that $\nu^*_t\leq \mathbb E[C^*_t|\mathcal F^0_t]$. Thus, given the admissible condition for consumption in Definition \ref{def:admissibility}, it is natural to assume that $\nu$ belongs to $\bigcap_{\beta\in\mathbb R}L^{\beta,T}_+$.  
\end{remark}

 
 \subsubsection{The fixed point condition}\label{sec:fixed-point}
By Definition \ref{def:equilibrium},
the fixed point condition reads
\begin{equation}\label{eq:fixed-point}
		\widehat\nu^*_t = \mathbb E[	\widehat c^*_t|\mathcal F^0_t		]+\mathbb E[	\widehat X^*_t|\mathcal F^0_t		],~0\leq t<T, \quad  \text{and }\quad \widehat\nu^*_T = \mathbb E[	\widehat X^*_T|\mathcal F^0_T	],
\end{equation}
where the process $c^*$ is determined in \eqref{consumption-optimization}. 
Moreover, for $0\leq t<T$,
in view of the optimality condition \eqref{consumption-optimization} the equilibrium consumption process satisfies 
 \begin{equation}\label{consumption-equilibrium-1}
 		c^*_t=\delta^\psi e^{  -\theta(\psi-1)\log\nu^*_t  } (\alpha e^{Y_t})^{ -\frac{\psi}{\vartheta} }=\delta^\psi e^{-\theta(\psi-1)\mathbb E[	\widehat c^*_t|\mathcal F^0_t	]} e^{-\theta(\psi-1)\mathbb E[\widehat X^*_t|\mathcal F^0_t]} (\alpha e^{Y_t})^{-\frac{\psi}{\vartheta}},
 \end{equation}
 which implies that the conditional expectation $\mathbb E[\widehat c^*_t|\mathcal F^0_t]$ satisfies the equation
 \begin{equation*}
 		\mathbb E[\widehat c^*_t|\mathcal F^0_t]=~\mathbb E[ \psi\log\delta	 ] - \mathbb E[\theta(\psi-1)]\mathbb E[	 \widehat c^*_t |\mathcal F^0_t	] - \mathbb E[\theta(\psi-1)]\mathbb E[\widehat X^*_t|\mathcal F^0_t] -  \mathbb E\left[\frac{\psi}{\vartheta}\log\alpha\right]  -  \mathbb E\left[ \left.\frac{\psi}{\vartheta} Y_t \right|\mathcal F^0_t\right].
\end{equation*}		
Thus, 
\[		
 	\mathbb E[\widehat c^*_t|\mathcal F^0_t]=~ \frac{1}{1+\mathbb E[\theta(\psi-1)]}\left\{       \mathbb E[ \psi\log\delta]-\mathbb E[ \theta(\psi-1) ]\mathbb E[\widehat X^*_t|\mathcal F^0_t]		- \mathbb E\left[ \frac{\psi}{\vartheta}\log\alpha  \right]	-\mathbb E\left[\left.  \frac{\psi}{\vartheta}Y_t  \right|\mathcal F^0_t \right]		\right\}
\]
and so
\begin{equation}\label{fixed-point-nu}
\begin{split}
	\widehat\nu^*_t & ~ =\mathbb E[\widehat X^*_t|\mathcal F^0_t] + \mathbb E[\widehat c^*_t|\mathcal F^0_t] \\
	& ~ =\frac{1}{1+\mathbb E[ \theta(\psi-1) ]}\left\{   \mathbb E[ \psi\log\delta ] - \mathbb E\left[ \frac{\psi}{\vartheta}\log\alpha   \right] - \mathbb E\left[ \left. \frac{\psi}{\vartheta} Y_t\right|\mathcal F^0_t   \right]        \right\} + \frac{1}{1+\mathbb E[ \theta(\psi-1)  ]}\mathbb E[ \widehat X^*_t|\mathcal F^0_t ].
\end{split}
\end{equation}
Define 	
\[
	\widetilde Y:=Y+\theta(1-\gamma)\mathbb E[\widehat X^*|\mathcal F^0].
\]
	This process satisfies the backward equation with zero terminal condition 
\begin{equation}\label{MF-BSDE-1}
	\begin{split}
		\widetilde Y_t =&~-\theta(1-\gamma)\int_t^T \mathbb E\left[\left.    \pi^*_sh_s-c^*_s-\frac{1}{2}(\pi^*_s)^2(\sigma^2_s+(\sigma^0_s)^2)	\right|\mathcal F^0_s	 \right]  \,ds \\
		&~+ \int_t^T  ~\left\{ \frac{Z^2_s+(Z^0_s)^2}{2} +\frac{1-\gamma}{2\gamma} \frac{	(h_s+\sigma_s Z_s+\sigma^0_sZ^0_s)^2		}{\sigma^2_s+(\sigma^0_s)^2}  \right. \\ 
		&~ \qquad \left. \qquad +\frac{1-\gamma}{\psi-1} \delta^\psi(\nu^*_s)^{-\theta(\psi-1)} (\alpha e^{Y_s})^{ -\frac{\psi}{\vartheta} } -\delta\vartheta\right\}\,ds\\
		&~ - \int_t^TZ_s\,dW_s-\int_t^T\left\{ Z^0_s+\theta(1-\gamma) \mathbb E[\pi^*_s\sigma^0_s |\mathcal F^0_s	] \right\} \,dW^0_s.
	\end{split}
\end{equation}
	To bring this equation into standard BSDE format we set
\begin{equation}\label{change:Z}
	\widetilde Z:=Z,\qquad \widetilde Z^0:=Z^0+\theta(1-\gamma)\mathbb E[\pi^*\sigma^0|\mathcal F^0]=Z^0+\theta(1-\gamma)\mathbb E\left[ \left.  \frac{h+\sigma Z+\sigma^0Z^0}{\gamma( \sigma^2+(\sigma^0)^2  )}\sigma^0   \right|\mathcal F^0       \right],
\end{equation}
where the second equality is equivalent to 
\begin{equation}\label{change:Z-0}
	Z^0=\widetilde Z^0-\frac{\theta(1-\gamma)\mathbb E\left[\left.     \frac{\sigma^0(  h+\sigma\widetilde Z+\sigma^0\widetilde Z^0   )}{\gamma(  \sigma^2+(\sigma^0)^2    )}		    \right|\mathcal F^0     \right]}{   1+\mathbb E\left[ \left.  	\frac{\theta(1-\gamma)(\sigma^0)^2}{\gamma( \sigma^2+(\sigma^0)^2  )}				  \right|\mathcal F^0   \right]   }.
\end{equation}
Moreover, using the definition of $\widetilde Y$, the equalities in \eqref{consumption-equilibrium-1} and \eqref{fixed-point-nu} imply that {on $[0,T)$}
\begin{equation}\label{equilibrium-consumption}
	\begin{split}
		c^* 
		=&~\delta^\psi \alpha^{ -\frac{  \psi }{\vartheta} }\exp\left(  -\theta(\psi-1)\widehat\nu^*-\frac{\psi}{\vartheta}Y							         \right)\\
		=&~\delta^\psi \alpha^{ -\frac{  \psi }{\vartheta} }\exp\left(  		-\frac{\theta(\psi-1)}{1+\mathbb E[ \theta(\psi-1) ]}\left\{   \mathbb E[\psi\log\delta]  - \mathbb E\left[ 	\frac{\psi}{\vartheta}\log\alpha	   \right]	    \right\}	+\frac{\theta(\psi-1)}{1+\mathbb E[\theta(\psi-1)]}	\mathbb E\left[ \left. \frac{\psi}{\vartheta}\widetilde Y	\right|\mathcal F^0	 \right]	-\frac{\psi}{\vartheta}\widetilde Y				         \right)\\
		=&~\exp\left(   -\frac{\theta(\psi-1)}{1+\mathbb E[ \theta(\psi-1) ]}\left\{        \mathbb E[\psi\log\delta]  - \mathbb E\left[ 	\frac{\psi}{\vartheta}\log\alpha	   \right]	  -\mathbb E\left[  \left.		\frac{\psi}{\vartheta}\widetilde Y \right|\mathcal F^0      \right]          \right\}		+\psi\log\delta-\frac{\psi}{\vartheta}\log\alpha-\frac{\psi}{\vartheta}\widetilde Y		          \right).
	\end{split}
\end{equation}

Taking the first equality in \eqref{change:Z} and \eqref{change:Z-0} into the first equality in \eqref{consumption-optimization}, we get the equilibrium investment process
     \begin{equation}\label{equilibrium-investment}
	\pi^*=\frac{ h+\sigma \widetilde Z + \sigma^0\widetilde Z^0	}{\gamma(	\sigma^2+(\sigma^0)^2	)}  -  \frac{	\sigma^0\theta(1-\gamma)\mathbb E\left[\left.     \frac{\sigma^0(  h+\sigma\widetilde Z+\sigma^0\widetilde Z^0   )}{\gamma(  \sigma^2+(\sigma^0)^2    )}		    \right|\mathcal F^0     \right]			}{\gamma(	\sigma^2+(\sigma^0)^2		)   \left(   1+\mathbb E\left[ \left.  	\frac{\theta(1-\gamma)(\sigma^0)^2}{\gamma( \sigma^2+(\sigma^0)^2  )}				  \right|\mathcal F^0   \right]  \right)  }
\end{equation}
Taking the first equality in \eqref{change:Z} and \eqref{change:Z-0} into \eqref{MF-BSDE-1}, we get 
   \begin{equation}\label{MF-BSDE:tilde-Y}
	\begin{split}
		\widetilde Y_t
		=&~  \int_t^T \Big\{ -\delta\vartheta +\theta(1-\gamma)\mathbb E[ c^*_s|\mathcal F^0_s	] + \frac{1-\gamma}{\psi-1}c^*_s\Big\}\,ds \\
		&~+\int_t^T\mathcal J_{\widetilde Z,\widetilde Z^0}(s)\,ds - \int_t^T\widetilde Z_s\,dW_s-\int_t^T\widetilde Z^0_s\,dW^0_s,
	\end{split}
\end{equation}
where $c^*$ is given by \eqref{equilibrium-consumption}, and $\mathcal J_{\widetilde Z,\widetilde Z^0}$ collects all terms containing $\widetilde Z$ and $\widetilde Z^0$, which is the same as \cite[Appendix B]{Fu2023}, with $\gamma$ there replaced by $1-\gamma$.

The following theorem summarizes the findings of this subsection.
\begin{theorem}\label{thm:BSDE-to-NE}
If the mean field BSDE \eqref{MF-BSDE:tilde-Y} admits a solution $(\widetilde Y,\widetilde Z, \widetilde Z^0)\in L^\infty\times H^2_{BMO}\times H^2_{BMO}$, then \eqref{equilibrium-consumption} and \eqref{equilibrium-investment} define candidate equilibrium rates $(\pi^*,c^*)\in H^2_{BMO}\times L^\infty_{\mathrm{log},+}$. 
If in addition, $(\pi^*,c^*)$ satisfies condition (iv) in $\mathcal A$, then they form a mean-field equilibrium.
\end{theorem}
 
 \begin{remark}
 	If $(\widetilde Z,\widetilde Z^0)\in L^\infty\times L^\infty$, which is the case in Section \ref{sec:existence}, then $(\pi^*,c^*)\in\mathcal A$. In particular, the dollar consumption $c^*X^*$ belongs to $\bigcap_{\beta\in\mathbb R}L^{\beta,T}_+$, thus, satisfies condition (iv) in $\mathcal A$.
 \end{remark}
 
 \begin{remark}
Theorem \ref{thm:BSDE-to-NE} focuses solely on equilibrium consumption and investment. The best response corresponding to a given $\nu$ can be analyzed in a similar manner.
 \end{remark}
 
\subsection{From Nash equilibrium to BSDE}\label{sec:necessary}

So far we derived a BSDE whose solution yields a mean-field equilibrium for our portfolio game. We proceed to prove that {\sl any} Nash equilibrium satisfies the previously obtained BSDE. To do so, it is crucial to establish a necessary characterization of the best response.

In what follows, we denote by $f_1$ and $f_2$ the derivatives of the aggregator $f$ w.r.t. the first and the second argument, respectively. In particular,  
\begin{equation}\label{eq:f_1}
	f_1(C,V) = \delta C^{-\frac{1}{\psi}}\left\{ (1-\gamma)V  \right\}^{ 1- \frac{1}{\vartheta} }
\end{equation}
and
\begin{equation}\label{eq:f_2}
	\begin{split}
	f_2(C,V)=&~\frac{	\delta C^{1-\frac{1}{\psi}}		}{1-\frac{1}{\psi}} \left(  1-\frac{1}{\vartheta} \right) (1-\gamma) \big((1-\gamma) V\big)^{-\frac{1}{\vartheta}}-\delta\vartheta \\
	=&~\delta C^{1-\frac{1}{\psi}}(\vartheta-1)\left\{  (1-\gamma)V    \right\}^{ -\frac{1}{\vartheta} } - \delta\vartheta.
	\end{split}
\end{equation}
For future reference, we denote the comparison-adjusted optimal consumption by
\begin{equation}\label{def:G*}
	G^* = c^*X^*\nu^{-\theta} = c^* e^{\widehat X^*}\nu^{-\theta}.
\end{equation}
Our necessary analysis is applied to the optimal log-wealth $\widehat X^*$, whose dynamics follows
\begin{equation}\label{eq:optimal-log-wealth}
    			d\widehat X^*_t = \left(  \pi^*_t h_t - c^*_t - \frac{1}{2}( \pi^*_t )^2\Sigma^2_t    \right)\,dt + \pi^*_t\sigma_t\,dW_t+\pi^*_t\sigma^0_t\,dW^0_t,\quad \widehat X^*_0=\log x.
\end{equation}
We start with the following necessary maximum principle for Epstein-Zin preferences.
\begin{proposition}\label{prop:maximum-principle}
Let $\nu\in\bigcap_{\beta\in\mathbb R}L^{\beta,T}_+$. 
	Let $(\pi^*,c^*):=(\pi^*(\nu),c^*(\nu))\in\mathcal A$ be an optimal control of \eqref{utility-max} and $(X^*,V^*):=(X^*(\nu),V^*(\nu))$ be the associated state process.  Define the following (well-posed) system of adjoint equations: 
	\begin{equation}\label{eq:adjoint}
		\left\{\begin{split}
			dQ_t=&~ f_2( G^*_t, V^*_t )Q_t\,dt,  \\
			-dP_t=&~ -Q_t f_1( G^*_t, V^*_t )G^*_t\,dt -K_t\,dW_t-K^0_t\,dW^0_t, \\
			 Q_0=&~-1,\quad P_T=-\alpha  Q_T e^{(1-\gamma)(\widehat X^*_T-\theta\widehat\nu_T)}.
		\end{split}\right.
	\end{equation}
	 Then, it holds that 
	\begin{equation}\label{optimality-1}
		Ph-P\Sigma^2\pi^*+\sigma K+ \sigma^0K^0=0.
	\end{equation}
and that
\begin{equation}\label{optimality-2}
	 Pc^*+Qf_1(G^*,V^*)G^*=0. 	
\end{equation}
\end{proposition}

\begin{proof}

	{\it Estimate for utility and adjoint variables.}
	We start with a series of a priori estimates.

We first consider the utility process. For any admissible pair $(\pi,c)$ it follows from  
		\cite[Theorem 4.6]{Seifried-2015} that
		\begin{equation}\label{power-bound}
			Y^{\gamma}	\leq	V^{\pi,c}\leq U_{\gamma}\circ U^{-1}_{\frac{1}{\psi}}( Y^{\frac{1}{\psi}} ),
		\end{equation}
		where (following the notation in \cite{Seifried-2015})
		\[
		Y^{\varrho}_t:=e^{\delta t}\mathbb E\left[\left. \int_t^T \delta e^{-\delta s} U_{\varrho}(C_s\nu^{-\theta}_s)\,ds +  e^{-\delta T}U_{\varrho}(C_T\nu^{-\theta}_T)   \right|\mathcal F_t   \right]   \quad   \textrm{ and }\quad U_{\varrho}(C):=\alpha\frac{C^{1-\varrho}}{1-\varrho}.          
		\] 
		
		Our assumptions on the consumption rate $c$ guarantee that the lower bound of the utility process $V^{\pi,c}$ belongs to $\bigcap_{\beta>0}  S^\beta $. Since 
		\[
		U_\gamma\circ U_{\frac{1}{\psi}}^{-1}(x)=\frac{\alpha^{1-\vartheta}}{1-\gamma} \left\{ \left(  1-\frac{1}{\psi}  \right)x\right\}^{  \vartheta }, 
		\]	
		the upper bound  also belongs to $\bigcap_{\beta>0} S^\beta $. Moreover, it follows from \eqref{power-bound}  that $(1-\gamma)V^{\pi,c}\geq 0$. Hence,
		\begin{equation}\label{bound:f2}
			f_2(C,V)\leq-\delta \vartheta. 
		\end{equation}
	 We proceed to bound the adjoint variables. By \eqref{eq:f_2}, the process $Q$ satisfies
		\begin{equation}\label{estimate:Q}
			0>Q_t=-\exp\left(  \int_0^t  f_2(  G^*_s, V^*_s    )\,ds   \right) \geq - e^{-\delta \vartheta t}. 
		\end{equation}
		By the admissible condition ($X^*_T=C_T^*$ and $C^*\in\bigcap_{\beta\in\mathbb R}L^{\beta,T}_+$), we deduce that 
		$
			P\in\bigcap_{\beta>0}S^\beta.
		$ 
As a result, $K\in\bigcap_{\beta>0}M^\beta$ and $K^0\in\bigcap_{\beta>0}M^\beta$.
        
	{\it First order variations and their estimate.}
	To obtain \eqref{optimality-1} and \eqref{optimality-2}, it is sufficient to perturbe the optimal strategy $(\pi^*,c^*)$ locally. For each $\epsilon>0$, each progressively measurable and {\it bounded} pair of perturbation $(\eta,\kappa)$, define
	\[
			\pi^\epsilon = \pi^*+\epsilon\eta,\qquad c^\epsilon = c^* e^{\epsilon\kappa}.
	\]
	In particular, $c^\epsilon$ is positive-valued and thus admissible. 
	Let $\widehat X^\epsilon$ be the log-wealth corresponding to $(\pi^\epsilon,c^\epsilon)$. Then,
	\begin{equation}\label{eq:log-wealth-epsilon}
			d\widehat X^\epsilon_t = \left(  \pi^\epsilon_t h_t - c^\epsilon_t - \frac{1}{2}( \pi^\epsilon_t )^2\Sigma^2_t    \right)\,dt + \pi^\epsilon_t\sigma_t\,dW_t+\pi^\epsilon_t\sigma^0_t\,dW^0_t,\quad \widehat X^\epsilon_0=\log x.
	\end{equation}
	By the boundedness of $\eta$, $\kappa$, $c^*$ and $\pi^*\in H^2_{BMO}$ as well as the SDE for $\widehat X^*$ in \eqref{eq:optimal-log-wealth}, the first order variation of log-wealth satisfies 
	\begin{equation}\label{variation-X}
		d\mathcal X_t = (\eta_t h_t-\kappa_t c^*_t-\eta_t\Sigma^2_t\pi^*_t )\,dt + \eta_t\sigma_t\,dW_t + \eta_t\sigma^0_t\,dW^0_t,\quad \mathcal X_0=0.
	\end{equation}
	where $\mathcal X:=\left.\frac{d}{d\epsilon}\widehat X^\epsilon \right\vert_{\epsilon=0}$. 
	
	By the BSDE representation for the recursive utilities $V^\epsilon$ and $V^*$, it holds 
	\begin{equation*}
		\begin{split}
			-d\frac{V^\epsilon_t - V^*_t}{\epsilon} =&~ \int_0^1 f_1\left( \lambda c^\epsilon_t e^{\widehat X^\epsilon_t} \nu^{-\theta}_t+(1-\lambda)	c^*_t e^{\widehat X^*_t}\nu^{-\theta}_t , \lambda V^\epsilon_t+(1-\lambda)V^*_t 	\right) \left(	\frac{c^\epsilon_t e^{\widehat X^\epsilon_t}-c^*_t e^{\widehat X^*_t}}{\epsilon}  \right)\nu^{-\theta}_t \,d\lambda  \,dt \\
			&~+ \int_0^1 f_2\left(     \lambda c^\epsilon_t e^{\widehat X^\epsilon_t} \nu^{-\theta}_t+(1-\lambda)	c^*_t e^{\widehat X^*_t}\nu^{-\theta}_t , \lambda V^\epsilon_t+(1-\lambda)V^*_t     \right)\left(\frac{V^\epsilon_t-V^*_t}{\epsilon}\right)\,d\lambda\,dt \\
			&~-\frac{\zeta^\epsilon_t - \zeta_t^*}{\epsilon} \,dW_t - \frac{\zeta^{0,\epsilon}_t - \zeta^{0,*}_t}{\epsilon} \,dW^0_t.
		\end{split}
	\end{equation*}
Using again $C^*\in\bigcap_{\beta\in\mathbb R}L^{\beta,T}_+$ and $\nu\in \bigcap_{\beta\in\mathbb R}L^{\beta,T}_+$, we can interchange the limit of $\epsilon\rightarrow 0$ and the integral to obtain 
\begin{equation}\label{variation-V}
	\left\{\begin{split}
	-dv_t =&~ \left\{  f_1(G^*_t,V^*_t)G^*_t(\mathcal X_t+\kappa_t ) + f_2(G^*_t,V^*_t) v_t \right\}\,dt - \zeta_t'\,dW_t - \zeta^{0\prime}_t\,dW^0_t, \\
	v_T=&~ \alpha \exp\left\{  (1-\gamma)( \widehat X^*_T-\theta\widehat\nu_T ) \right\}\mathcal X_T.
	\end{split}\right.
\end{equation}
Finally, $C^*\in\bigcap_{\beta\in\mathbb R}L^{\beta,T}_+$ and $\nu\in \bigcap_{\beta\in\mathbb R}L^{\beta,T}_+$ imply that
\[
		\mathcal X \in \bigcap_{\beta>0} S^\beta, \quad v\in \bigcap_{\beta>0} S^\beta,\quad \zeta'\in \bigcap_{\beta>0} M^{\beta},\quad \zeta^{0\prime} \in \bigcap_{\beta>0} M^\beta. 
\]
 {\it Integration by parts}. 
We apply an integration by parts to obtain that  
\begin{equation}\label{MP-int-by-part:1}
	\begin{split}
&~\alpha Q_T\exp\left\{ (1-\gamma)( \widehat X^*_T-\theta\widehat\nu_T)   \right\}\mathcal X_T\\
	=&~ v_TQ_T  = -v_0-\int_0^T Q_t f_1(G^*_t,V^*_t)G^*_t(\mathcal X_t+\kappa_t)\,dt + \int_0^T Q_t\zeta'_t\,dW_t + \int_0^TQ_t\zeta^{0\prime}_t\,dW^0_t 
	\end{split}
\end{equation}
and that 
\begin{equation}\label{MP-int-by-part:2}
	\begin{split}
		& ~  -\mathcal X_TQ_T\alpha e^{(1-\gamma)(\widehat X^*_T-\theta\widehat\nu_T)}  =\mathcal X_TP_T\\
		=&~\int_0^T \left\{	P_t(\eta_th_t - \kappa_t c^*_t -\eta_t\Sigma^2_t\pi^*_t  ) + \mathcal X_t Q_t f_1( G^*_t, V^*_t )G^*_t + \eta_t\sigma_tK_t + \eta_t \sigma^0_tK^0_t	\right\} \,dt  \\
		&~+ \int_0^T\left\{   P_t\eta_t\sigma_t+\mathcal X_t K_t		 \right\}\,dW_t + \int_0^T\left\{    	 P_t\eta_t\sigma^0_t+\mathcal X_t K^0_t			   \right\}\,dW^0_t .
	\end{split}
\end{equation}
The above estimates imply that all stochastic integrals in \eqref{MP-int-by-part:1} and \eqref{MP-int-by-part:2} are true martingales.
Summing up \eqref{MP-int-by-part:1} and \eqref{MP-int-by-part:2} shows that
\begin{equation*}
	\begin{split}
		0\geq v_0 =  \mathbb E\left[   \int_0^T \eta_t \left\{	P_t( h_t - \Sigma^2_t\pi^*_t  ) +  \sigma_tK_t +   \sigma^0_tK^0_t	\right\} \,dt  -\int_0^T \left\{  Q_t f_1(G^*_t,V^*_t)G^*_t  +  P_t   c^*_t  \right\}\kappa_t  \,dt      \right],
	\end{split}
\end{equation*} 
where the inequality is because of optimality.

  {\it Conclusion.} 
By choosing 
\[
	\eta_t=(	P_t( h_t - \Sigma^2_t\pi^*_t  ) +  \sigma_tK_t +   \sigma^0_tK^0_t ) \mathbf{1}_{  \{	|P_t( h_t - \Sigma^2_t\pi^*_t  ) +  \sigma_tK_t +   \sigma^0_tK^0_t| \leq m		 \}  } 
\]
and  
\[
	\kappa_t=-\left(Q_t f_1(G^*_t,V^*_t)G^*_t  +  P_t   c^*_t  \right)\mathbf{1}_{ \{  | Q_t f_1(G^*_t,V^*_t)G^*_t  +  P_t   c^*_t   |\leq m  \} }, 
\]
it holds that 
\begin{equation*}
	\begin{split}
	&~\left(	P_t( h_t - \Sigma^2_t\pi^*_t  ) +  \sigma_tK_t +   \sigma^0_tK^0_t		\right)^2 \mathbf{1}_{  \{	| P_t( h_t - \Sigma^2_t\pi^*_t  ) +  \sigma_tK_t +   \sigma^0_tK^0_t	| \leq m		 \}  } \\
	&~+ \left(Q_t f_1(G^*_t,V^*_t)G^*_t  +  P_t   c^*_t   \right)^2\mathbf{1}_{ \{  |Q_t f_1(G^*_t,V^*_t)G^*_t  +  P_t   c^*_t   |\leq m  \} }=0.
	\end{split}
\end{equation*}
Letting $m\rightarrow\infty$ yields \eqref{optimality-1} and \eqref{optimality-2}.
\end{proof}



The following corollary defines a consumption-adjusted adjoint variable $\overline P$, whose representation turns out to be useful in proving the one-to-one correspondence. 
\begin{corollary}\label{coro:overline-P}
	Define $\overline P_t=e^{\int_0^t c^*_r\,dr} P_t$, where $P$ is the adjoint variable in Proposition \ref{prop:maximum-principle}. Then there exists $(\lambda,\lambda^0)$ such that
	\begin{equation}\label{representation:overline-P}
			\overline P_t = \mathcal E\left(  \int_0^t \lambda_s\,dW_s + \int_0^t \lambda^0_s\,dW^0_s   \right).
	\end{equation}
As a result, 
\begin{equation}\label{representation:P}
	dP_t = -c^*_t P_t\,dt +P_t\lambda_t\,dW_t+P_t\lambda^0_t\,dW^0_t.
\end{equation}
\end{corollary}
\begin{proof}
	Taking \eqref{optimality-2} into the BSDE in \eqref{eq:adjoint}, we have 
	\[
		-dP_t=c^*_t P_t\,dt - K_t\,dW_t - K^0_t\,dW^0_t.
	\]
	The definition of $\overline P$ yields 
	\[
		d\overline P_t = e^{\int_0^t c^*_s\,ds} K_t\,dW_t + e^{\int_0^t c^*_s\,ds} K^0_t\,dW^0_t.
	\]
	Since $\overline P_T=-\alpha e^{\int_0^T c^*_s\,ds}  Q_T e^{(1-\gamma)(\widehat X^*_T-\theta\widehat\nu_T)}>0$ by \eqref{estimate:Q}, we know that $\overline P$ is a positive martingale. Positive martingale representation implies \eqref{representation:overline-P}. Thus,
	\[
			d\overline P_t = \overline P_t\lambda_t\,dW_t + \overline P_t\lambda^0_t\,dW^0_t,
	\]
	which implies $P\lambda=K$ and $P\lambda^0=K^0$. Taking them into the BSDE in \eqref{eq:adjoint} we obtain \eqref{representation:P}.
\end{proof}


We are now ready to show that any optimal control satisfies the BSDE with the driver \eqref{driver-MOP}. As a result, any equilibrium strategy must satisfy mean field BSDE \eqref{MF-BSDE:tilde-Y} determined in Section \ref{sec:fixed-point}.

\begin{theorem}\label{thm:necessary}
	Let $(\lambda,\lambda^0)\in H^2_{BMO}$, where $(\lambda,\lambda^0)$ is defined in Corollary \ref{coro:overline-P}.
	Let $(\pi^*,c^*)\in\mathcal A$ be a best response to $\nu\in \bigcap_{\beta\in\mathbb R}L^{\beta,T}_+$, and $(X^*,V^*)$ be the corresponding state process. Then, 
	 \begin{equation}\label{necessary-characterization}
		\begin{split}
			\pi^* = \frac{ h+\sigma Z+\sigma^0 Z^0			}{\gamma(	\sigma^2+(\sigma^0)^2	)} ,\quad 
			c^*  =  \delta^\psi (\nu^{-\theta})^{\psi-1}(\alpha e^Y)^{-\frac{\psi}{\vartheta}},
		\end{split}
	\end{equation}
where $(Y,Z,Z^0)$ satisfies the BSDE with the driver \eqref{driver-MOP} and terminal value $-\theta(1-\gamma)\log\nu_T$.
\end{theorem}
\begin{proof}
	
Define
\begin{equation}\label{def:Y-Z-Z0}
	\left\{\begin{split}
		Y =&~ \log P - \log\alpha - (1-\gamma)\widehat X^* - \log(-Q),\\
		Z = &~ \lambda - (1-\gamma)\pi^*\sigma,\\
		Z^0=&~ \lambda^0 - (1-\gamma)\pi^*\sigma^0.
	\end{split}\right. 
\end{equation}
We will verify $(Y,Z,Z^0)$ satisfies the desired BSDE.

The first equality in \eqref{def:Y-Z-Z0} implies $Y_T=-\theta(1-\gamma)\widehat\nu_T$.

Corollary \ref{coro:overline-P} (especially $K=P\lambda$ and $K^0=P\lambda^0$) and \eqref{optimality-1} imply 
\[
	Ph-P\Sigma^2\pi^*+\sigma P\lambda+\sigma^0 P\lambda^0=0 \quad \Rightarrow \quad  h - \Sigma^2\pi^*+\sigma\lambda + \sigma^0\lambda^0=0.
\]
Taking the second and the third equalities in \eqref{def:Y-Z-Z0} into the above equation yields 
\begin{equation*}
	h - \Sigma^2\pi^* + \sigma(	Z+(1-\gamma)\pi^*\sigma	) + \sigma^0(	Z^0+(1-\gamma)\pi^*\sigma^0	)=0,
\end{equation*}
which implies 
\begin{equation}\label{relation:pi*-candidate-Z}
		\pi^* = \frac{h+\sigma Z+\sigma^0Z^0}{\gamma\Sigma^2}. 
\end{equation}
From \eqref{eq:adjoint}, \eqref{representation:P} and \eqref{def:Y-Z-Z0}, we have $-dY_t = g_t\,dt-Z_t\,dW_t-Z^0_t\,dW^0_t$, where
\[
		g = \gamma c^* + f_2(G^*,V^*) + \frac{1}{2}(\lambda^2+(\lambda^0)^2) + (1-\gamma)\pi^* h -\frac{1}{2}(1-\gamma)(\pi^*)^2\Sigma^2.
\] 
Using \eqref{def:Y-Z-Z0} and \eqref{relation:pi*-candidate-Z}, the last three terms in $g$ satisfy
\[
		\frac{1}{2}\left( \lambda^2+(\lambda^0)^2   \right)^2 + (1-\gamma)\pi^* h - \frac{1}{2}(1-\gamma)(\pi^*)^2\Sigma^2= \frac{1}{2}\left( Z^2+(Z^0)^2   \right) + \frac{1-\gamma}{2\gamma} \frac{(h+\sigma Z+\sigma^0Z^0)^2}{\Sigma^2}.
\]
Thus,
\begin{equation}\label{candidate-driver}
	g=\gamma c^*+f_2(G^*,V^*) + \frac{1}{2}\left( Z^2+(Z^0)^2   \right) + \frac{1-\gamma}{2\gamma} \frac{(h+\sigma Z+\sigma^0Z^0)^2}{\Sigma^2}.
\end{equation}
Our goal is to prove \eqref{candidate-driver} coincides with \eqref{driver-MOP}. To do so, we will first prove  
\begin{equation}\label{def:mathcal-V}
	\mathcal V: = \frac{\alpha}{1-\gamma}e^{(1-\gamma)\widehat X^*+Y} = V^*
\end{equation}
and then establish
\begin{equation}\label{relation:c*-candidate-Y}
	c^* = \delta^\psi\nu^{-\theta(\psi-1)}(\alpha e^Y)^{-\frac{\psi}{\vartheta}}.
\end{equation}
We emphasize that $(\pi^*,c^*)$ in \eqref{relation:pi*-candidate-Z} and \eqref{relation:c*-candidate-Y} do not coincide with \eqref{necessary-characterization} before we prove $(Y,Z,Z^0)$ satisfies the desired BSDE.

By the definition of $\mathcal V$ in \eqref{def:mathcal-V}, 
\begin{equation}\label{terminal:mathcal-V-V*}
		\mathcal V_T = \frac{\alpha}{1-\gamma} e^{(1-\gamma)\widehat X^*_T-\theta(1-\gamma)\widehat\nu_T }=\frac{\alpha}{1-\gamma} e^{(1-\gamma)\log X^*_T+(1-\gamma)\log\nu^{-\theta}_T} = \alpha\frac{(	X^*_T\nu^{-\theta}_T	)^{1-\gamma}}{1-\gamma} = V^*_T.
\end{equation}
Using dynamics of $Y$ (see \eqref{candidate-driver} for its driver) and $\widehat X^*$ (see \eqref{eq:optimal-log-wealth}), we have 
\begin{equation}\label{BSDE:mathcal-V}
	\begin{split}
		-d\mathcal V_t = \mathcal V_t\left(    	c^*_t+f_2(G^*_t,V^*_t)								    \right)\,dt - \mathcal V_t\lambda_t\,dW_t - \mathcal V_t\lambda^0_t\,dW^0_t.
	\end{split}
\end{equation}
In order to compare the driver of $\mathcal V$ and the driver of $V^*$, we consider rewriting the latter one, $f(G^*,V^*)$. At this stage, the second optimality condition \eqref{optimality-2} becomes instrumental. On the one hand, the definition of $Y$ in \eqref{def:Y-Z-Z0} and the definition of $\mathcal V$ in \eqref{def:mathcal-V} imply that 
\begin{equation}\label{eq:mathcal-V-PQ}
	\begin{split}
	e^Y = \frac{1}{\alpha} \left(\frac{P}{-Q}\right) e^{-(1-\gamma)\widehat X^*} \quad \Rightarrow \quad  (1-\gamma)\mathcal V =  \alpha  e^{Y+(1-\gamma)\widehat X^*} = -\frac{P}{Q}.
	\end{split}
\end{equation}
On the other hand, the second optimality condition \eqref{optimality-2} implies 
\begin{equation}\label{eq:f1-G*-1}
	f_1(G^*,V^*) G^* = c^*\left( -\frac{P}{Q} \right) = (1-\gamma)c^*\mathcal V,
\end{equation}
where the second equality is given by \eqref{eq:mathcal-V-PQ}.
By \eqref{eq:f_1}, 
\begin{equation}\label{eq:f1-G*-2}
	f_1(G^*,V^*)G^* = \delta (G^*)^{1-\frac{1}{\psi}}\left\{   	(1-\gamma)V^*	  \right\}^{1-\frac{1}{\vartheta}}.
\end{equation}
Combining \eqref{eq:f1-G*-1} and \eqref{eq:f1-G*-2}, we have 
\begin{equation}\label{eq:c*-mathcal-V}
	c^*\mathcal V = \frac{\delta}{1-\gamma}(G^*)^{1-\frac{1}{\psi}} \left\{     (1-\gamma)V^*   \right\}^{ 1-\frac{1}{\vartheta} } = \delta(G^*)^{1-\frac{1}{\psi}} V^*\left\{  (1-\gamma)V^* \right\}^{-\frac{1}{\vartheta}}.
\end{equation}
By \eqref{eq:f_2} and \eqref{eq:c*-mathcal-V}, we have
\begin{equation}\label{eq:f-f2}
		f(G^*,V^*) = c^*\mathcal V+f_2(G^*,V^*)V^*.
\end{equation}
Taking \eqref{eq:f-f2} into the BSDE for $V^*$ and considering the difference $\mathcal V-V^*$, we have 
\begin{equation}\label{BSDE:mathcal-V-V*}
	-d(\mathcal V_t - V^*_t) = f_2(G^*_t,V^*_t)(	\mathcal V_t - V^*_t )\,dt + \text{local martingale},
\end{equation}
whose unique solution is zero by \eqref{terminal:mathcal-V-V*} and Lemma \ref{lemma:BSDE-zero}, stated below. Thus, \eqref{def:mathcal-V} is proved. 

To prove \eqref{relation:c*-candidate-Y}, we take $\mathcal V=V^*$ into \eqref{eq:c*-mathcal-V} to obtain 
\begin{equation}\label{relation:c*-V*}
	c^*=\delta (G^*)^{1-\frac{1}{\psi}} \left\{  (1-\gamma)V^* \right\}^{-\frac{1}{\vartheta}}.
\end{equation}
Taking the definiton of $G^*$ in \eqref{def:G*} into \eqref{relation:c*-V*} gives \eqref{relation:c*-candidate-Y}.

Finally, by \eqref{relation:c*-V*} and \eqref{eq:f-f2}, we have
\[
	(\vartheta-1)c^* = f_2(G^*,V^*)+\delta\vartheta, 
\]
which implies that 
\[
		\gamma c^* + f_2(G^*,V^*) = \gamma c^*+(\vartheta-1)c^*-\delta\vartheta =\frac{1-\gamma}{\psi-1} c^* - \delta\vartheta. 
\]
Taking the above equality into \eqref{candidate-driver}, we have
\[
	g = \frac{1-\gamma}{\psi-1}c^* - \delta\vartheta + \frac{1}{2}(Z^2+(Z^0)^2) + \frac{1-\gamma}{2\gamma} \frac{(h+\sigma Z+\sigma^0Z^0)^2}{\Sigma^2},
\]
which coincides with \eqref{driver-MOP} by using \eqref{relation:c*-candidate-Y}.
 \end{proof}
The BSDE \eqref{BSDE:mathcal-V-V*} is a linear BSDE, however, its driver is only bounded from above (see \eqref{bound:f2}). Thus, standard result in e.g. \cite[Proposition 6.2.1]{Pham-2008} does not immediately imply that the solution is identical to zero. We need the following lemma.
\begin{lemma}\label{lemma:BSDE-zero}
	Let assumptions in Theorem \ref{thm:necessary} hold. The unique solution to \eqref{BSDE:mathcal-V-V*} is zero. 
\end{lemma}
\begin{proof}
	{First, we prove the sequence $\{ Q_\tau\mathcal V_\tau, \text{ stopping time }\tau  \}$ is uniformly integrable, where $Q$ is the forward adjoint variable in \eqref{eq:adjoint}.} Indeed, define \[ 
		\overline{\mathcal V}_t=\exp\left\{   \int_0^t  c^*_s + f_2(G^*_s,V^*_s)	\,ds	 \right\}\mathcal V_t. 
	\]
	Then it satisfies $d\overline{\mathcal V}_t=\overline{\mathcal V}_t\lambda_t\,dW_t+\overline{\mathcal V}_t\lambda^0_t\,dW^0_t$, which implies $\overline{\mathcal V}_t=\overline{\mathcal V}_0\mathcal E\left(   \int_0^t\lambda_s\,dW_s + \int_0^t\lambda^0_s\,dW^0_s     \right)$. By assumptions in Theorem \ref{thm:necessary}, in particular, $(\lambda,\lambda^0)\in H^2_{BMO}$, the process $\overline{\mathcal V}$ satisfies reverse H\"older's inequality for some $p>1$ (see \cite[Theorem 3.1]{Kazamaki-2006}). Consequently, $\{ \overline{\mathcal V}_\tau, \text{ stopping time }\tau \}$ is uniformly integrable. Since $ c^*\in L^\infty_{\text{log},+}$, $\{ Q_\tau\mathcal V_\tau, \text{ stopping time }\tau  \}$ is also uniformly integrable.
	
	Next, let $\tau_n\uparrow T$ be a sequence of localization stopping times such that the local martingale in \eqref{BSDE:mathcal-V-V*} is a true martingale. We have 
	\[
			\mathbb E[Q_{\tau_n} ( \mathcal V_{\tau_n}	-V^*_{\tau_n}	) |\mathcal F_t ] = Q_t(\mathcal V_t-V^*_t).
	\]
	Since $V^*\in\bigcap_{\beta>0}S^\beta$, $Q$ is bounded by \eqref{estimate:Q} and $Q\mathcal V$ is uniformly integrable, we have
	\[
			Q_t(\mathcal V_t-V^*_t)=\lim_{n\rightarrow \infty}\mathbb E[Q_{\tau_n} ( \mathcal V_{\tau_n}	-V^*_{\tau_n}	) |\mathcal F_t ]=\mathbb E[Q_{T} ( \mathcal V_{T}	-V^*_{T}	) |\mathcal F_t ] = 0,
	\]
	which yields the desired result.
\end{proof}

 By Theorem \ref{thm:BSDE-to-NE} and Theorem \ref{thm:necessary}, 
 we have the following one-to-one correspondence between the equilibrium investment and consumption rates and the solution to the BSDE \eqref{MF-BSDE:tilde-Y}. 
    \begin{theorem}\label{thm:one-to-one}
    	Define
    	\[
    		\mathscr C_{adm} = \{  (\widetilde Y,\widetilde Z,\widetilde Z^0)\in L^\infty\times H^2_{BMO}\times H^2_{BMO}:   (\widetilde Y,\widetilde Z,\widetilde Z^0) \text{ is a solution to } \eqref{MF-BSDE:tilde-Y}     \text{ and the generated }(\pi^*,c^*)\in\mathcal A          \}
    	\]
    	and 
    	\[
    		\mathscr E_{adm} = \{	(\pi^*,c^*)\in \mathcal A: (\pi^*,c^*)	\text{ is an equilibrium rate and the generated }(\lambda,\lambda^0)\in H^2_{BMO}							\}.
    	\]
  Then, there is a one-to-one correspondence between $\mathscr C_{adm}$ and $\mathscr E_{adm}$.

The relation is given by
     \begin{equation}\label{pi-Z-Z0-widetilde}
   		\pi^*=\frac{ h+\sigma \widetilde Z + \sigma^0\widetilde Z^0	}{\gamma(	\sigma^2+(\sigma^0)^2	)}  -  \frac{	\sigma^0\theta(1-\gamma)\mathbb E\left[\left.     \frac{\sigma^0(  h+\sigma\widetilde Z+\sigma^0\widetilde Z^0   )}{\gamma(  \sigma^2+(\sigma^0)^2    )}		    \right|\mathcal F^0     \right]			}{\gamma(	\sigma^2+(\sigma^0)^2		)   \left(   1+\mathbb E\left[ \left.  	\frac{\theta(1-\gamma)(\sigma^0)^2}{\gamma( \sigma^2+(\sigma^0)^2  )}				  \right|\mathcal F^0   \right]  \right)  }
   \end{equation}
   and
   \begin{equation}\label{c-Y-widetilde}
   	\left\{\begin{split}
   		c^*
   		=&~\exp\left(   -\frac{\theta(\psi-1)}{1+\mathbb E[ \theta(\psi-1) ]}\left\{        \mathbb E[\psi\log\delta]  - \mathbb E\left[ 	\frac{\psi}{\vartheta}\log\alpha	   \right]	  -\mathbb E\left[  \left.		\frac{\psi}{\vartheta}\widetilde Y \right|\mathcal F^0      \right]          \right\}	\right. \\
		& \qquad \qquad \left.	+\psi\log\delta-\frac{\psi}{\vartheta}\log\alpha-\frac{\psi}{\vartheta}\widetilde Y		          \right)\quad \text{on }[0,T),\\
		c^*_T=&~1.
   	\end{split}\right.
   \end{equation}
\end{theorem}

  \begin{proof} 
  	
  	$\mathscr C_{adm}\Rightarrow \mathscr E_{adm}$. For any $(\widetilde Y,\widetilde Z,\widetilde Z^0)\in \mathscr C_{adm}$, Theorem \ref{thm:BSDE-to-NE} yields $(\pi^*,c^*)$ defined by \eqref{pi-Z-Z0-widetilde} and \eqref{c-Y-widetilde}. By \eqref{def:Y-Z-Z0}, $(\lambda,\lambda^0)\in H^2_{BMO}\times H^2_{BMO}$. Thus, $(\pi^*,c^*)\in\mathscr E_{adm}$.

$\mathscr E_{adm}\Rightarrow \mathscr C_{adm}$. For any $(\pi^*,c^*)\in \mathscr E_{adm}$, by Remark \ref{remark:integrability-nu}, $\nu^*\in\bigcap_{\beta\in\mathbb R}L^{\beta,T}_+$. Thus, by Theorem \ref{thm:necessary} and equilibrium condition \eqref{eq:fixed-point}, $(\pi^*,c^*)\in\mathcal A$ must be characterized by the solution to the following mean field FBSDE system
 	 \begin{equation}\label{MF-FBSDE}
 		\left\{\begin{split}
 			dX_t^{*}=&~X^{*}_t\Big(\pi^*_t h_t	\,dt+\pi^*_t\sigma_t\,dW_t	+\pi^*_t\sigma^0_t\,dW^0_t		\Big) - c^*_tX^{*}_t\,dt\\
 			-dY_t=&~\left\{ \frac{Z^2_t+(Z^0_t)^2}{2} + \frac{1-\gamma}{2\gamma} \frac{	(h_t+\sigma_t Z_t+\sigma^0_tZ^0_t)^2		}{\sigma^2_t+(\sigma^0_t)^2} \right.  \\
 			&~\left.+\frac{1-\gamma}{\psi-1} \delta^\psi(\nu^*_t)^{-\theta(\psi-1)} (\alpha e^{Y_t})^{ -\frac{\psi}{\vartheta} } -\delta\vartheta\right\}\,dt - Z_t\,dW_t-Z^0_t\,dW^0_t\\
 			X^*_0=&~x,\quad Y_T=-\theta(1-\gamma)\mathbb E[\widehat X^*_T|\mathcal F^0_T],
 		\end{split}\right. 
 	\end{equation}
 	where 
 	\begin{equation*}\label{equilibrium-pi-c}
 		\begin{split}
 			\pi^*=\frac{ h+\sigma Z+\sigma^0 Z^0			}{\gamma(	\sigma^2+(\sigma^0)^2	)},\qquad c^*=\delta^\psi (\nu^*)^{-\theta(\psi-1)}(\alpha e^Y)^{-\frac{\psi}{\vartheta}},
 		\end{split}
 	\end{equation*}
 	and
 	\begin{equation*}\label{eq:log-nu*}
 	\widehat\nu^*=\frac{1}{1+\mathbb E[ \theta(\psi-1) ]}\left\{   \mathbb E[ \psi\log\delta ] - \mathbb E\left[ \frac{\psi}{\vartheta}\log\alpha   \right] - \mathbb E\left[ \left. 		\frac{\psi}{\vartheta} Y\right|\mathcal F^0   \right]        \right\} + \frac{1}{1+\mathbb E[ \theta(\psi-1)  ]}\mathbb E[ \widehat X^*|\mathcal F^0 ].
 	\end{equation*}
 	By the same argument from \eqref{eq:fixed-point} to \eqref{MF-BSDE:tilde-Y}, $(\pi^*,c^*)$ satisfies \eqref{pi-Z-Z0-widetilde} and \eqref{c-Y-widetilde}, with $(\widetilde Y,\widetilde Z,\widetilde Z^0)$ satisfies \eqref{MF-BSDE:tilde-Y}. 
 	
 	The relation \eqref{c-Y-widetilde} implies that 
 	\[
 		\widetilde Y= 	-\frac{\vartheta}{\psi}\log c^*-\frac{\theta\vartheta(\psi-1)}{\psi}\mathbb E[  \log c^* |\mathcal F^0 ]+\vartheta\log\delta-\log\alpha,			\quad \text{on }[0,T)
 	\]
 	which together with $\widetilde Y_T=0$ implies $\widetilde Y\in L^\infty$ since $ c^*\in L^\infty_{\text{log},+}$. Moreover, by \eqref{def:Y-Z-Z0} and \eqref{change:Z}-\eqref{change:Z-0}, we have $(\widetilde Z,\widetilde Z^0)\in H^2_{BMO}\times H^2_{BMO}$. Thus, $(\widetilde Y,\widetilde Z,\widetilde Z^0)\in \mathscr C_{adm}$. 
\end{proof}

 \begin{remark}\label{rem:reverse}
 	The BMO condition for $(\lambda,\lambda^0)$ is consistent with the reverse H\"older inequality in \cite{ET-2015,FR-2011,Fu-2021,Fu2023}, where time additive utility (exponential or power utility) is considered. In general, this condition is necessasry to establish the one-to-one correspondence between the equilibrium investment rate and the $Z$-component of solution to some (F)BSDE in the BMO space. This condition can be dropped if either the following two conditions holds: 
	\begin{itemize}
 		\item If only common noise exists, i.e. $\sigma=0$, then \eqref{pi-Z-Z0-widetilde} yields a one-to-one correspondence between $\pi^*$ and $\widetilde Z^0$ (and thus $Z^0$), which implies that $\pi^*\in H^2_{BMO}$ is equivalent to $\widetilde Z^0\in H^2_{BMO}$, even without the BMO condition for $\lambda$. This also implies the equivalence between $\pi^*\in L^\infty$ and $Z^0\in L^\infty$.
 	\item  If all coefficients are deterministic, then $\widetilde Z=\widetilde Z^0=0$ and $\pi^*\in H^2_{BMO}$ trivially holds, if it exists.
\end{itemize}
 \end{remark}
To further demonstrate the power of our approach, in Section \ref{sec:power} we revisit the power utility case studied in \cite{Fu-2021,Fu2023}. Before doing so, we complete the uniqueness analysis and the $N$-player game analysis in the next two sections.  
  
\section{Uniqueness of mean-field equilibrium strategy in $H^2_{BMO}\times L^\infty$}\label{sec:uniqueness}
 
 By Theorem \ref{thm:one-to-one}, to establish our uniqueness of equilibrium result, it is sufficient and necessary to show that the BSDE \eqref{MF-BSDE:tilde-Y} admits at most one solution. 

\begin{theorem}\label{thm:uniqueness} 
	For each $R>0$, there exists at most one equilibrium investment and consumption rates $(\pi^*,c^*)\in H^2_{R,BMO}\times L^\infty$ such that the generated $(\lambda,\lambda^0)\in H^2_{BMO}\times H^2_{BMO}$, when the competition parameter $\theta$ is small enough. Here, $H^2_{R, BMO}$ is the $R$-ball of $H^2_{BMO}$.
%
\end{theorem}
\begin{proof}
		Let $(\widetilde Y,\widetilde Z,\widetilde Z^0)$ and $(\widetilde Y',\widetilde Z',\widetilde Z^{0\prime})$ be solutions to the BSDE \eqref{MF-BSDE:tilde-Y} in $L^\infty\times H^2_{BMO}\times H^2_{BMO}$ and let
	\[
	\Delta Y:=\widetilde Y-\widetilde Y', \quad \Delta Z:=\widetilde Z-\widetilde Z', \quad \Delta Z^0:=\widetilde Z^0-\widetilde Z^{0\prime}. 
	\]
It follows that
\begin{equation*}
	\begin{split}
		\Delta Y_t=\theta(1-\gamma)\int_t^T\mathbb E[ \Delta c^*_s |\mathcal F^0_s ]\,ds + \int_t^T\frac{1-\gamma}{\psi-1}  \Delta c^*_s  \,ds+\int_t^T \Delta\mathcal J_s \,ds-\int_t^T\Delta Z_s\,dW_s-\int_t^T\Delta Z^0_s\,dW^0_s,
	\end{split}
\end{equation*}
where 
\begin{equation*}
	\Delta c^*=H \left( \frac{\theta(\psi-1)}{1+\mathbb E[\theta(\psi-1)]}\mathbb E\left[\left. \frac{\psi}{\vartheta}\Delta Y \right|\mathcal F^0   \right] - \frac{\psi}{\vartheta}\Delta Y \right)
\end{equation*}
with $H$ being bounded since $\widetilde Y,\widetilde Y'\in L^\infty$, and
where 
{\scriptsize  
	\begin{align*}
			\Delta \mathcal J =&~ -\theta\gamma\mathbb E\left[\left.  f^{\sigma h}\Delta Z+f^{\sigma^0h}\Delta Z^0\right|\mathcal F^0   \right]+\theta\gamma\mathbb E\left[\left. \theta\gamma f^{\sigma^0h} \right|\mathcal F^0    \right]  \frac{	\mathbb E\left[ \left.  f^{\sigma^0\sigma}\Delta Z+f^{\sigma^0\sigma^0}\Delta Z^0  \right|\mathcal F^0 	\right]				}{1+\mathbb E[\theta\gamma f^{\sigma^0\sigma^0}|\mathcal F^0]}\\
			&~+\theta\gamma\mathbb E\left[ 	\frac{1}{2}(1-\gamma)\Sigma^2 \left\{    2f^h+f^\sigma(	\widetilde Z+\widetilde Z'	) + f^{\sigma^0}(	\widetilde Z^0+\widetilde Z^{0\prime}	) -\frac{	\theta\gamma f^{\sigma^0}\mathbb E\left[ 2f^{\sigma^0 h}+f^{\sigma^0\sigma}(	\widetilde Z+\widetilde Z'	)+f^{\sigma^0\sigma^0}(	\widetilde Z^0+\widetilde Z^{0\prime}		)	|\mathcal F^0    \right]				}{1+\mathbb E[\theta\gamma f^{\sigma^0\sigma^0}|\mathcal F^0  ]} 							\right\} \right.\\
			&~\qquad\quad \left. \left. \times 	\left\{ f^\sigma\Delta Z+f^{\sigma^0}\Delta Z^0  - \frac{	\theta\gamma f^{\sigma^0}\mathbb E\left[\left.	 f^{\sigma^0\sigma} \Delta Z+f^{\sigma^0\sigma^0}\Delta Z^0 \right|\mathcal F^0 \right]		}{ 1+\mathbb E\left[\left.	\theta\gamma f^{\sigma^0\sigma^0} \right|\mathcal F^0\right]			}	    \right\}	\right|\mathcal F^0	    \right] \\
						&~-\frac{1}{2}\left\{   \widetilde Z^0+\widetilde Z^{0\prime} - \frac{	\theta\gamma\mathbb E\left[ \left. 2f^{\sigma^0h}+f^{\sigma^0\sigma}(	\widetilde Z+\widetilde Z'	)+f^{\sigma^0\sigma^0}(  \widetilde Z^0+\widetilde Z^{0\prime}   ) 			     \right|\mathcal F^0  \right]						}{1+\mathbb E\left[\left. \theta\gamma f^{\sigma^0\sigma^0} \right|\mathcal F^0 \right]}					\right\}   \frac{			\theta\gamma\mathbb E\left[\left.	f^{\sigma^0\sigma}\Delta Z+f^{\sigma^0\sigma^0}\Delta Z^0\right|\mathcal F^0			\right]	}{1+\mathbb E\left[\left.  \theta\gamma f^{\sigma^0\sigma^0}  \right|\mathcal F^0 \right]}      \\
			&~-\frac{\gamma(1-\gamma)\Sigma^2}{2}\left\{ 	2f^h+f^\sigma(	\widetilde Z+\widetilde Z'	) +f^{\sigma^0}( \widetilde Z^0+\widetilde Z^{0\prime} ) - \frac{	\theta\gamma f^{\sigma^0}\mathbb E\left[\left.	2f^{\sigma^0h}+f^{\sigma^0\sigma}(\widetilde Z+\widetilde Z' )+f^{\sigma^0\sigma^0}(\widetilde Z^0+\widetilde Z^{0\prime})	\right|\mathcal F^0	\right]				}{1+\mathbb E\left[ \left. \theta\gamma f^{\sigma^0\sigma^0} \right|\mathcal F^0  \right]}		      \right\}\\
			&~\qquad\qquad\quad\quad  \times  \frac{ \theta\gamma f^{\sigma^0}\mathbb E\left[\left.  f^{\sigma^0\sigma}\Delta Z+f^{\sigma^0\sigma^0}\Delta Z^0  \right|\mathcal F^0  \right]  }{1+\mathbb E\left[\left. \theta\gamma f^{\sigma^0\sigma^0}  \right|\mathcal F^0  \right]} \\
			&~+\frac{\widetilde Z+\widetilde Z'}{2}\Delta Z   \\
			&~+\frac{\gamma(1-\gamma)\Sigma^2}{2}\left\{ 	2f^h+f^\sigma(	\widetilde Z+\widetilde Z'	) +f^{\sigma^0}( \widetilde Z^0+\widetilde Z^{0\prime} ) - \frac{	\theta\gamma f^{\sigma^0}\mathbb E\left[\left.	2f^{\sigma^0h}+f^{\sigma^0\sigma}(\widetilde Z+\widetilde Z' )+f^{\sigma^0\sigma^0}(\widetilde Z^0+\widetilde Z^{0\prime})	\right|\mathcal F^0	\right]				}{1+\mathbb E\left[ \left. \theta\gamma f^{\sigma^0\sigma^0} \right|\mathcal F^0  \right]}		      \right\}    f^\sigma\Delta Z  \\ 
						&~+\frac{\gamma(1-\gamma)\Sigma^2}{2}\left\{ 	2f^h+f^\sigma(	\widetilde Z+\widetilde Z'	) +f^{\sigma^0}( \widetilde Z^0+\widetilde Z^{0\prime} ) - \frac{	\theta\gamma f^{\sigma^0}\mathbb E\left[\left.	2f^{\sigma^0h}+f^{\sigma^0\sigma}(\widetilde Z+\widetilde Z' )+f^{\sigma^0\sigma^0}(\widetilde Z^0+\widetilde Z^{0\prime})	\right|\mathcal F^0	\right]				}{1+\mathbb E\left[ \left. \theta\gamma f^{\sigma^0\sigma^0} \right|\mathcal F^0  \right]}		      \right\}   f^{\sigma^0}\Delta Z^0 \\ 
			&~		+\frac{1}{2}\left\{   \widetilde Z^0+\widetilde Z^{0\prime} - \frac{	\theta\gamma\mathbb E\left[ \left. 2f^{\sigma^0h}+f^{\sigma^0\sigma}(	\widetilde Z+\widetilde Z'	)+f^{\sigma^0\sigma^0}(  \widetilde Z^0+\widetilde Z^{0\prime}   ) 			     \right|\mathcal F^0  \right]						}{1+\mathbb E\left[\left. \theta\gamma f^{\sigma^0\sigma^0} \right|\mathcal F^0 \right]}					\right\}   \Delta Z^0					\\
			:=&~\theta\Delta\widetilde{\mathcal J}  + \mathcal C\Delta Z+\mathcal C^0\Delta Z^0
	\end{align*}
}
with $f^a$ and $f^{ab}$ defined in \cite[Appendix B]{Fu2023}.

All terms in the definition of $\Delta\mathcal J$ that do not contain $\theta$ are linear terms of $(\Delta Z,\Delta Z^0)$ and hence can be dropped by a change of measure:
\[
		\frac{d\mathbb Q}{d\mathbb P}=\mathcal E\left\{ \int_0^\cdot \mathcal C_{s}		  \,dW_s+\int_0^\cdot	\mathcal C^0_s	 \,dW^0_s  \right\}.
\]
Since $\widetilde Z$, $\widetilde Z^0$, $\widetilde Z'$ and $\widetilde Z^{0\prime}$ are in $H^2_{BMO}$, $(W^{\mathbb Q}, W^{0,\mathbb Q})$ is a $\mathbb Q$-Brownian motion, where
\begin{equation*}
	\left\{\begin{split}
		W^{\mathbb Q}_\cdot =&~ W_\cdot-\int_0^\cdot \mathcal C_s\,ds,  \\
		W^{0,\mathbb Q}_\cdot =&~ W^0_\cdot - \int_0^\cdot \mathcal C^0_s\,ds.
	\end{split}\right. 
\end{equation*} 
As a result, 
\begin{equation*}
	\begin{split}
		\Delta Y_t = \theta(1-\gamma)\int_t^T\mathbb E[ \Delta c^*_s |\mathcal F^0_s ]\,ds + \int_t^T\frac{1-\gamma}{\psi-1}  \Delta c^*_s  \,ds+\theta\int_t^T \Delta\mathcal{\widetilde J}_s \,ds-\int_t^T\Delta Z_s\,dW_s^{\mathbb Q}-\int_t^T\Delta Z^0_s\,dW^{0,\mathbb Q}_s.
	\end{split}
\end{equation*}
It\^o's formula and standard estimate imply that  
\begin{equation*}
	\begin{split}
		&~\esssup_{\omega,t\leq s\leq T}(\Delta Y_s)^2 +  \| \Delta Z\|^2_{BMO,\mathbb Q} +  	\| \Delta Z^0\|_{BMO,\mathbb Q}^2\\
		\leq&~ C\int_t^T \esssup_{\omega,t\leq r\leq s}(\Delta Y_r)^2   \,ds + \theta  C\| \Delta Z\|^2_{BMO,\mathbb Q} + \theta C	\| \Delta Z^0\|_{BMO,\mathbb Q}^2,
	\end{split}
\end{equation*}
where $C$ depends on $R$.
It implies $\Delta Y=\Delta Z=\Delta Z^0=0$ by Gr\"onwall's inequality and letting $\theta$ be small enough.
\end{proof}

%
%
%



\section{Wellposedness under deterministic parameters}\label{sec:existence}

 This section proves the existence of an equilibrium in closed form for models with deterministic parameters\footnote{They may depend on an initial distribution capturing initial heterogeneity.}. Moreover, we verify that this closed form equilibrium investment rate and consumption rate is the unique one in $L^\infty\times L^\infty_{\text{log},+}$, without additional integrability assumptions. 
 Both the MFG and the $N$-player game are considered.

\subsection{The MFG}

{\bf Uniqueness in $L^\infty\times L^\infty_{\mathrm{log},+}$.}  	Let $(\widetilde Y,\widetilde Z,\widetilde Z^0)$ and $(\widetilde Y',\widetilde Z',\widetilde Z^{0\prime})$ be solutions to the BSDE \eqref{MF-BSDE:tilde-Y} in $L^\infty\times L^\infty\times L^\infty$ and let
\[
\Delta Y:=\widetilde Y-\widetilde Y', \quad \Delta Z:=\widetilde Z-\widetilde Z', \quad \Delta Z^0:=\widetilde Z^0-\widetilde Z^{0\prime}. 
\]	
The triple $(\Delta Y,\Delta Z,\Delta Z^{0})$ satisfies the BSDE
\begin{equation*}
	\begin{split}
		\Delta Y_t=&~\int_t^T \left(A_{1,s}\mathbb E\left[\left. \frac{\psi}{\vartheta} \Delta Y_s   \right|\mathcal F^0_s   \right]+ A_{2,s}\Delta Y_s\right)\,ds +\int_t^T\left( A_{3,s}\mathbb E[A_{4,s} \Delta Z_s|\mathcal F^0_s] + A_{5,s}\mathbb E[A_{6,s} \Delta Z^0_s|\mathcal F^0_s] \right)\,ds\\
		&~+\int_t^T \left( A_{7,s}\Delta Z_s+A_{8,s}\Delta Z^0_s  \right) \,ds  -\int_t^T\Delta Z_s\,dW_s-\int_t^T\Delta Z^0_s\,dW^0_s,
	\end{split}
\end{equation*}
where all coefficients $A_i$ belong to $L^\infty$ since $(\widetilde Y,\widetilde Z,\widetilde Z^0)$ and $(\widetilde Y',\widetilde Z',\widetilde Z^{0\prime})$ are assumed to be in $L^\infty\times L^\infty\times L^\infty$. Standard estimates show that 
\begin{equation*}
	\begin{split}
		&~\mathbb E\left[  (\Delta Y_t)^2 \right]+\mathbb E\left[ \int_t^T (	\Delta Z_s)^2+(\Delta Z^0_s)^2\,ds     \right]\\
		=&~2\mathbb E\left[\int_t^T \left( A_{1,s}\Delta Y_s\mathbb E\left[\left. \frac{\psi}{\vartheta} \Delta Y_s   \right|\mathcal F^0_s   \right]+A_{2,s}(\Delta Y_s)^2 + A_{3,s}\Delta Y_s\mathbb E\left[ A_{4,s}\Delta Z_s|\mathcal F^0_s   \right]+A_{5,s}\Delta Y_s\mathbb E[A_{6,s}\Delta Z^0_s|\mathcal F^0_s]	\right)		\,ds  \right]\\
		&~+2\mathbb E\left[  \int_t^T   \Delta Y_s  \left( A_{7,s}\Delta Z_s+A_{8,s}\Delta Z^0_s  \right) \,ds   \right]\\
		\leq&~	C\mathbb E\left[  \int_t^T (	\Delta Y_s	)^2\,ds  \right]+		\frac{1}{2}\mathbb E\left[ \int_t^T (	\Delta Z_s	)^2+(  \Delta Z^0_s  )^2 \,ds  \right],
	\end{split}
\end{equation*}
which implies by Gr\"onwall's inequality $\Delta Y=\Delta Z=\Delta Z^0=0$.

Note that $(\pi,c)\in L^\infty\times L^\infty_{\text{log},+}$ implies that $(\pi,c)\in\mathcal A$. Moreover, recalling the second point in Remark \ref{rem:reverse}, we know that there exists at most one equilibrium investment rate and consumption rate in $L^\infty\times L^\infty_{\text{log},+}$.

{\bf Closed form solution.}
If the model parameters are deterministic, then $\widetilde Z=\widetilde Z^0=0$, which implies that 
\[
	Z=0,\qquad Z^0=-\frac{\theta(1-\gamma)\mathbb E\left[  \frac{\sigma^0 h}{\gamma(\sigma^2+(\sigma^0)^2)}   \right]}{	1+\mathbb E\left[  \frac{\theta(1-\gamma)(\sigma^0)^2}{ \gamma(  \sigma^2+(\sigma^0)^2  )  }    \right]		},
\]	
and
\[	
	\pi^*=\frac{h}{\gamma(\sigma^2+(\sigma^0)^2)}  -\frac{\sigma^0}{\gamma( \sigma^2+(\sigma^0)^2  )}\times \frac{\theta(1-\gamma)\mathbb E\left[  \frac{\sigma^0 h}{\gamma(\sigma^2+(\sigma^0)^2)}   \right]}{	1+\mathbb E\left[  \frac{\theta(1-\gamma)(\sigma^0)^2}{ \gamma(  \sigma^2+(\sigma^0)^2  )  }    \right]		} \in L^\infty. 
\]
Taking the above equalities into the driver of $\widetilde Y$, we get that
\begin{equation}\label{eq:Y-widetilde}
	\begin{split}
		\widetilde Y_t=&~-\theta(1-\gamma)\int_t^T \mathbb E\left[    	\pi^*_sh_s-c^*_s-\frac{1}{2}(\pi^*_s)^2(\sigma^2_s+(\sigma^0_s)^2)			    \right]  \,ds \\
				&~+ \int_t^T  ~\left\{ \frac{ (Z^0_s)^2}{2} +\frac{1-\gamma}{2\gamma} \frac{	(h_s +  \sigma^0_sZ^0_s)^2		}{\sigma^2_s+(\sigma^0_s)^2} +\frac{1-\gamma}{\psi-1} c^*_s -\delta\vartheta\right\}\,ds\\
				=&~\int_t^T	\left\{   \theta(1-\gamma)\mathbb E[c^*_s] +\frac{1-\gamma}{\psi-1}c^*_s	  \right\}  		\,ds\\
				&~-\theta(1-\gamma)\int_t^T\mathbb E\left[ 	\pi^*_sh_s-\frac{1}{2}(\pi^*_s)^2(\sigma^2_s+(\sigma^0_s)^2)	    \right]\,ds\\
				&~+\int_t^T   \left\{ \frac{ (Z^0_s)^2}{2} +\frac{1-\gamma}{2\gamma} \frac{	(h_s +  \sigma^0Z^0)^2		}{\sigma^2+(\sigma^0)^2}   -\delta\vartheta\right\}\,ds.
	\end{split}
\end{equation}
We now show that the above equation can be reduced to a Riccati equation. To this end,  
we set 
\begin{equation}\label{eq:A}
	\begin{split}
		A:=\frac{(Z^0)^2}{2}+\frac{1-\gamma}{2\gamma} \frac{		(h+\sigma^0Z^0)^2			}{\sigma^2+(\sigma^0)^2}-\delta\vartheta - \theta(1-\gamma)\mathbb E\left[  \pi^*h-\frac{1}{2}(\pi^*)^2(\sigma^2+(\sigma^0)^2)		   \right].
	\end{split}
\end{equation}
Moreover, from \eqref{c-Y-widetilde} we get that 
\begin{equation}\label{c-Y-hat}
	\begin{split}
		c^*
		=&~\exp\left(   -\frac{\theta(\psi-1)}{1+\mathbb E[ \theta(\psi-1) ]}\left\{        \mathbb E[\psi\log\delta]  - \mathbb E\left[ 	\frac{\psi}{\vartheta}\log\alpha	   \right]	  -\mathbb E\left[  \left.		\frac{\psi}{\vartheta}\widetilde Y \right|\mathcal F^0      \right]          \right\}		+\psi\log\delta-\frac{\psi}{\vartheta}\log\alpha-\frac{\psi}{\vartheta}\widetilde Y		          \right)\\
		=&~\exp\left(   \frac{\theta(\psi-1)}{1+\mathbb E[ \theta(\psi-1) ]}     \mathbb E\left[  \left.		\breve Y \right|\mathcal F^0      \right]           	- \breve Y		          \right),
	\end{split}
\end{equation}
where
\[
    \breve Y:= -\psi\log\delta+\frac{\psi}{\vartheta}\log\alpha+\frac{\psi}{\vartheta}\widetilde Y.	     
\]
Expressing the optimal consumption plan $c^*$ in terms of $\breve Y$ as shown in \eqref{c-Y-hat} and recalling the equation \eqref{eq:Y-widetilde} we see that  
\begin{equation*}
	\begin{split}
		\breve Y' &~ = -\frac{\psi}{\vartheta}   \theta(1-\gamma) \mathbb E\left[ \exp\left(   \frac{\theta(\psi-1)}{1+\mathbb E[ \theta(\psi-1) ]}\mathbb E[\breve Y]- \breve Y         \right)\right]  -\frac{\psi}{\vartheta}\frac{1-\gamma}{\psi-1} \exp\left(   \frac{\theta(\psi-1)}{1+\mathbb E[ \theta(\psi-1) ]}\mathbb E[\breve Y]- \breve Y         \right)      -\frac{\psi}{\vartheta}A\\
		& = ~ -   \theta(\psi-1) \mathbb E\left[ \exp\left(   \frac{\theta(\psi-1)}{1+\mathbb E[ \theta(\psi-1) ]}\mathbb E[\breve Y]- \breve Y         \right)\right]  -  \exp\left(   \frac{\theta(\psi-1)}{1+\mathbb E[ \theta(\psi-1) ]}\mathbb E[\breve Y]- \breve Y         \right)      -\frac{\psi}{\vartheta}A.\\
	\end{split}
\end{equation*} 
Taking expectations we get that
 \begin{equation*}
  		\mathbb E[\breve Y]'= -(\mathbb E\left[   \theta(\psi-1)\right]+1 )\mathbb E\left[ \exp\left(   \frac{\theta(\psi-1)}{1+\mathbb E[ \theta(\psi-1) ]}\mathbb E[\breve Y]- \breve Y         \right)\right]       -\mathbb E\left[\frac{\psi}{\vartheta}A \right]
\end{equation*}
from which we conclude that
\begin{align*}
\frac{		\theta(\psi-1)	}{ \mathbb E\left[   \theta(\psi-1)\right]+1 } \mathbb E[\breve Y]' - \breve Y'= \exp\left(   \frac{\theta(\psi-1)}{1+\mathbb E[ \theta(\psi-1) ]}\mathbb E[\breve Y]- \breve Y         \right) -  \frac{ \theta(\psi-1)}{ \mathbb E\left[   \theta(\psi-1)\right]+1  }\mathbb E\left[\frac{\psi}{\vartheta}A \right]+\frac{\psi}{\vartheta}A.
 \end{align*}
 
 Let us now put 
 \[
 	\mathring Y:=\frac{		\theta(\psi-1)	}{ \mathbb E\left[   \theta(\psi-1)\right]+1 } \mathbb E[\breve Y] - \breve Y 
	\quad \mbox{and} \quad \bar Y=\exp(\mathring Y).
 \]
Then, 
\begin{equation*}
	\mathring Y'=e^{\mathring Y} -  \frac{ \theta(\psi-1)}{ \mathbb E\left[   \theta(\psi-1)\right]+1  }\mathbb E\left[\frac{\psi}{\vartheta}A \right]+\frac{\psi}{\vartheta}A
\end{equation*}
and so the function $\bar Y$ satisfies the Riccati equation 
 \[
 	\bar Y'=\bar Y^2+\left(  -  \frac{ \theta(\psi-1)}{ \mathbb E\left[   \theta(\psi-1)\right]+1  }\mathbb E\left[\frac{\psi}{\vartheta}A \right]+\frac{\psi}{\vartheta}A  \right)\bar Y
 \]
 with the terminal condition
 \begin{equation}\label{eq:D}
 		\bar Y_T=\exp\left( 	\frac{\theta(\psi-1)}{\mathbb E[\theta(\psi-1)]+1}\mathbb E\left[   -\psi\log\delta+\frac{\psi}{\vartheta}\log\alpha			  \right]	-\left(	-\psi\log\delta+\frac{\psi}{\vartheta}\log\alpha		   \right)					      \right):=D.
 \end{equation}
 
 The unique solution of the above Riccati equation is given by
 \[
 	\bar Y_t=D\left\{     \exp\left( \int_t^T B_r\,dr  \right) + D\int_t^T\exp\left(  \int_t^s B_r\,dr  \right)\,ds         \right\}
 \]
 where
 \begin{equation}\label{eq:B}
 		B=\left(  -  \frac{ \theta(\psi-1)}{ \mathbb E\left[   \theta(\psi-1)\right]+1  }\mathbb E\left[\frac{\psi}{\vartheta}A \right]+\frac{\psi}{\vartheta}A  \right).
 \end{equation}
 Thus, we have shown the following result, which includes \cite[Theorem 2.6]{Riedel-2024} as a special case.
 \begin{theorem}
 	Assume $h$, $\sigma$ and $\sigma^0$ are deterministic,
 	the {\it unique} equilibrium investment and consumption plan have the following closed form expression
 	\[
 			\pi^*_t=\frac{h_t}{\gamma(\sigma^2_t+(\sigma^0_t)^2)}  -\frac{\sigma^0_t}{\gamma( \sigma^2_t+(\sigma^0_t)^2  )}\times \frac{\theta(1-\gamma)\mathbb E\left[  \frac{\sigma^0_t h_t}{\gamma(\sigma^2_t+(\sigma^0_t)^2)}   \right]}{	1+\mathbb E\left[  \frac{\theta(1-\gamma)(\sigma^0_t)^2}{ \gamma(  \sigma^2_t+(\sigma^0_t)^2  )  }    \right]		},\quad t\in[0,T]
 	\]
 	and
 	\begin{equation*}
 		\left\{\begin{split}
 		c^*_t=&~D\left\{     \exp\left( \int_t^T B_r\,dr  \right) + D\int_t^T\exp\left(  \int_t^s B_r\,dr  \right)\,ds         \right\},\quad t\in[0,T),\\
 		c^*_T=&~1,
 		\end{split}\right. 
 \end{equation*}
 	where $A$, $D$ and $B$ are given by \eqref{eq:A}, \eqref{eq:D} and \eqref{eq:B}. 
 	
 	{As a special case, the simple equilibrium strategy obtained in \cite{Riedel-2024} is unique in $L^\infty\times L^\infty_{\mathrm{log},+}$.}
 \end{theorem}
 
 
 \subsection{The $N$-player game}
Having solved the MFG, solving $N$-player games requires only minor modifications of previously given arguments. In the $N$-player game we set 
 	$$
 	\nu = \left( \prod_{j \neq i} C^j \right)^{\frac{1}{N-1}}.
 	$$
A lengthy yet relatively straightforward computation yields the equilibrium investment strategy
\begin{equation*}
	\begin{split}
		\pi^{i,*}_t=\frac{h^i_t}{	\gamma^i(\sigma^i)^2+\left( 	\gamma^i-\frac{\theta^i(1-\gamma^i)}{N-1}		  \right)(\sigma^{i0}_t)^2			} - \frac{\theta^i (1-\gamma^i)\sigma^{i0}_t}{ \gamma^i(\sigma^{i}_t)^2 + \left( \gamma^i-\frac{\theta^i(1-\gamma^i)}{N-1}  \right)(\sigma^{i0}_t)^2  } \frac{\phi^N_t}{1+\psi^N_t},\quad t\in[0,T]
	\end{split}
\end{equation*}
where 
\[
	\phi^N:=\frac{1}{N-1}\sum_{j=1}^N \frac{	 h^j\sigma^{j0}		}{\gamma^j(\sigma^j)^2+\left(   \gamma^j -\frac{\theta^j(1-\gamma^j)}{N-1} \right)(	 \sigma^{j0}	)^2	}
\]
and
\[
	\psi^N:=\frac{1}{N-1}\sum_{j=1}^N \frac{\theta^j(1-\gamma^j)(\sigma^{j0})^2}{   \gamma^j(\sigma^{j})^2+\left(  \gamma^j-\frac{\theta^j(1-\gamma^j)}{N-1}  \right)(\sigma^{j0})^2      },
\]
and the equilibrium consumption plan
\[
	c^{i,*}_t=D^i\left\{	\exp\left( 	\int_t^T B^i_r\,dr  \right)	  + D^i\int_t^T\exp\left(  \int_t^s B^i_r\,dr	 \right)\,ds	\right\}, ~t\in[0,T), \text{ and } c^{i,*}_T=1,
\]
where the coefficients $B^i$ and $D^i$ are given by
\[
		B^i=b^iA^i-\frac{a^i}{1+\frac{1}{N-1}\sum_{i=1}^N a^i} \frac{1}{N-1}\sum_{i=1}^N b^iA^i,
\]
\[
	b^i=\frac{\frac{\psi^i-1}{1-\gamma^i}}{		1-\frac{\theta^i(\psi^i-1)}{N-1}			},\quad a^i=\frac{\theta^i(\psi^i-1)}{	1-\frac{\theta^i(\psi^i-1)}{N-1}		},
\]
\begin{equation*}
	\begin{split}
		A^i=&~-\theta^i(1-\gamma^i)\frac{1}{N-1}\sum_{j\neq i}\left\{   \pi^{j,*}h^j-\frac{1}{2}\left( (\sigma^j)^2+(\sigma^{j0})^2		\right)(\pi^{j,*})^2   \right\}+\frac{ (Z^{i0})^2}{2}+\frac{1}{2}\sum_{j\neq i}(Z^{ij})^2\\
		&~+\frac{1-\gamma^i}{2\gamma^i} \frac{  ( h^i  +  \sigma^{i0}Z^{i0}  )^2  }{ (\sigma^i)^2+(\sigma^{i0})^2  } -\delta^i\vartheta^i,
	\end{split}
\end{equation*}
\begin{equation*}
	\begin{split}
	Z^{i0} = & -\frac{\frac{\theta^{i} (1-\gamma^{i})}{N-1}}{1 - \frac{\theta^{i} (1-\gamma^{i}) (\sigma^{i0})^2}{(N-1) \gamma^{i}   \left((\sigma^{j})^2 + (\sigma^{j0})^2 \right)}} 
	\sum_{j \neq i} \frac{\sigma^{j0} h^j}{ \gamma^j \left\{ (\sigma^j)^2 + (\sigma^{j0})^2 \right\}} \\
	& + \frac{\frac{\theta^{i} (1-\gamma^{i})}{N-1}}{1 - \frac{\theta^{i} (1-\gamma^{i}) (\sigma^{i0})^2}{(N-1)\gamma^{i} \left((\sigma^{j})^2 + (\sigma^{j0})^2 \right)}} 
	\cdot \frac{1}{1 + \psi^{N}} \sum_{i=1}^{N} 
	\frac{\frac{\theta^{i} (1-\gamma^{i}) (\sigma^{i0})^2}{(N-1) \gamma^{i}  \left((\sigma^{j})^2 + (\sigma^{j0})^2 \right)}}{1 - \frac{\theta^{i} (1-\gamma^{i}) (\sigma^{i0})^2}{(N-1)   \gamma^{i}   \left((\sigma^{j})^2 + (\sigma^{j0})^2 \right)}} 
	\sum_{j \neq i} \frac{\sigma^{j0} h^j}{  \gamma^j   \left\{ (\sigma^j)^2 + (\sigma^{j0})^2 \right\}},
	\end{split}
\end{equation*}
and
\[
		Z^{ij}=-\frac{		\theta^i\gamma^i				}{	N-1	}\frac{\sigma^{i0}Z^{i0}+h^i}{ 	\gamma^i((\sigma^i)^2+(\sigma^{i0})^2	)	 }. 
\]

\section{The case with power utility revisited}\label{sec:power}

In this section, we apply our maximum-principle-based approach to the settings considered in \cite{Fu-2021, Fu2023}, where mean-field portfolio games with and without consumption under power utility, 
$$
U(x)=\frac{x^{1-\gamma}}{1-\gamma}, \quad \gamma>0, \gamma\neq 1,
$$
were studied. In those works, a one-to-one correspondence was established between equilibrium strategies and the solutions of certain BSDEs.
The strength of our approach lies in the fact that, when the Epstein-Zin utility reduces to power utility, (i) we recover the one-to-one correspondence under exactly the same assumptions as in \cite{Fu-2021}, without relying on the DPP in \cite{ET-2015}; and (ii) we derive the one-to-one correspondence in a larger space than that considered in \cite{Fu2023}. In particular, in both cases, the space $L^{\beta,T}_+$, which is specifically designed for Epstein--Zin utility, is not required.
 
To establish the one-to-one correspondence, the crucial step is the necessary characterization of the best response in MFGs. In this section, we restrict our attention to proving this necessary condition, since the sufficiency follows from arguments analogous to those in \cite{Fu-2021, Fu2023}. 
 
 \subsection{Power utility with only investment}\label{sec:power-investment}
 In this section, we revisit the problem in \cite{Fu-2021}: for a fixed $\nu$, 
 \[
 		J(\pi)=\mathbb E\left[U(X^\pi_T\nu_T^{-\theta})\right]\rightarrow \max_\pi,
 \]
 where the wealth process satisfies 
 \[
 		dX^\pi_t = \pi_t X^\pi_t\left( h_t\,dt + \sigma_t\,dW_t + \sigma^0_t\,dW^0_t   \right),\quad X^\pi_0=x.
 \]
 The log-wealth process satisfies 
 \[
 		d\widehat X^\pi_t = \left( \pi_t h_t - \frac{1}{2}\pi^2_t\Sigma^2_t \right)\,dt + \pi_t\sigma_t\,dW_t+\pi_t\sigma^0_t\,dW^0_t,\quad \widehat X^\pi_0=\log x.
 \]
 Let $\pi^*$ be the best response and $X^*$ be the optimal wealth process.   
 Define the conditional expected utility as
 \[
     D^*_t = \mathbb E\left[U(X^{*}_T \nu^{-\theta}_T  ) | \mathcal F_t		\right] = \mathbb E\left[ \left. \frac{1}{1-\gamma}\exp\left\{   (1-\gamma) ( \widehat X^*_T -\theta \widehat \nu_T)  \right\}  \right| \mathcal F_t    \right].
 \]
 Throughout Section \ref{sec:power-investment}, we make the following assumption. 
 \begin{ass}\label{ass:power-investment}
      	$D^* \text{ satisfies reverse H\"older inequality }R_b\text{ for some }b>1$; refer to e.g. \cite[Appendix B]{Fu-2021} for the definition of reverse H\"older inequality $R_b$.
 \end{ass}
In this section, we will show that in the power utility case, $\pi^*\in H^2_{BMO}$ alone, together with Assumption \ref{ass:power-investment}, is enough to establish the local maximum principle, and thus the one-to-one correspondence. In particular, we do not need $X^*_T$ to be integrable to all powers as in Definition \ref{def:admissibility} for the Epstein-Zin utility.

For $\epsilon>0$ and a progressively measurable bounded process $\eta$, define $\pi^\epsilon=\pi^*+\epsilon\eta$. The corresponding log-wealth $\widehat X^\epsilon$ satisfies
\begin{equation}\label{diff:log-wealth}
	\widehat X^\epsilon_t = \widehat X^*_t+\epsilon\xi^\eta_t-\frac{1}{2}\epsilon^2 \int_0^t \eta^2_s\Sigma^2_s\,ds, 
\end{equation}
where 
\begin{equation}\label{eq:xi-eta-t}
\xi^\eta_t = \int_0^t \eta_s(h_s - \pi^*_s\Sigma^2_s)\,ds + \int_0^t \eta_s\sigma_s\,dW_s + \int_0^t \eta_s\sigma^0_s\,dW^0_s.
\end{equation}
The next lemma establishes the first order variation of the expected utility.
 \begin{lemma}
 	
  Let $\pi^*\in H^2_{BMO}$. Define $P_T=e^{(1-\gamma)(\widehat X^*_T-\theta\widehat\nu_T)}$. Then it holds
   \begin{equation}\label{variation:J}
		  \lim_{\epsilon\rightarrow 0} \frac{J(\pi^\epsilon)-J(\pi^*)}{\epsilon} = \mathbb E[P_T\xi^\eta_T].
  \end{equation}
 \end{lemma}
 \begin{proof}

 	By \eqref{diff:log-wealth}, we have 
 	\[
 			\frac{J(\xi^\epsilon)-J(\pi^*)}{\epsilon} = \mathbb E\left[   \frac{e^{ (1-\gamma)(\widehat X^\epsilon-\theta\widehat\nu_T) } - e^{(1-\gamma)(\widehat X^*_T-\theta\widehat\nu_T)} }{\epsilon(1-\gamma)} \right] = \mathbb E\left[  \frac{P_T}{1-\gamma} \frac{ \exp\left\{(1-\gamma)(\epsilon\xi^\eta_T-\frac{1}{2}\epsilon^2\int_0^T\eta^2_t\Sigma^2_t\,dt)\right\} -1 }{\epsilon}    \right].
 	\]
 	Note that the $\epsilon$-dependent term inside the expectation admits a pointwise limit 
 	\[
 		\lim_{\epsilon\rightarrow 0}	\frac{1}{1-\gamma}\frac{ 	\exp\left\{(1-\gamma)(\epsilon\xi^\eta_T-\frac{1}{2}\epsilon^2\int_0^T\eta^2_t\Sigma^2_t\,dt)\right\} -1 	}{\epsilon} = \xi^\eta_T.
 	\]
To justify interchanging the limit and the expectation, we show that the random variable inside the expectation can be bounded by an integrable random variable that is independent of $\epsilon$. 	

 	Since $\pi^*\in H^2_{BMO}$, John-Nirenberg inequality (see \cite[Theorem 2.1]{Kazamaki-2006}) implies that for some positive constant $\varepsilon$ small enough 
 	\begin{equation}\label{John-Nirenberg}
		\mathbb E\left[  \exp\left\{  \varepsilon \int_0^T ( \pi^*_t	)^2\Sigma^2_t\,dt		\right\}     \right]<\infty.  		
 	\end{equation}
 For any constant $\beta$, we have by \eqref{John-Nirenberg}
 \begin{equation*}
 	\mathbb E\left[  \exp\left\{ \beta\left|  \int_0^T \eta_t\pi^*_t\Sigma^2_t	\,dt	\right|	\right\}     \right] \leq \text{const}(\beta,\varepsilon) \mathbb E\left[  \exp\left\{     \varepsilon \int_0^T	(\pi^*_t)^2\Sigma^2_t\,dt    \right\}   \right]<\infty,
 \end{equation*}
which implies that for any constant $\beta$
\begin{equation*}
	\mathbb E\left[  e^{\beta|\xi^\eta_T|}   \right]<\infty.
\end{equation*}
Using the inequality $|e^x-1|\leq |x|e^{|x|}$, we have 
\begin{equation*}
	\begin{split}
		&~\frac{ \exp\left\{(1-\gamma)(\epsilon\xi^\eta_T-\frac{1}{2}\epsilon^2\int_0^T\eta^2_t\Sigma^2_t\,dt)\right\} -1 }{\epsilon}   \\
		\leq&~ |(1-\gamma)   \xi^\eta_T |     \exp\left\{(1-\gamma)(\epsilon\xi^\eta_T-\frac{1}{2}\epsilon^2\int_0^T\eta^2_t\Sigma^2_t\,dt)\right\}.
	\end{split}
\end{equation*}
Since $\pi^*\in H^2_{BMO}$ and $\eta$ is bounded, $\xi^\eta_T$ is integrable to all powers. Moreover, since $D^*$ satisfies $R_b$, there exists some $r>1$ such that $P_T\in L^r$. Thus, for $\epsilon<1$
\[
		  \frac{P_T}{1-\gamma} \frac{ \exp\left\{(1-\gamma)(\epsilon\xi^\eta_T-\frac{1}{2}\epsilon^2\int_0^T\eta^2_t\Sigma^2_t\,dt)\right\} -1 }{\epsilon}   \leq |P_T  \xi^\eta_T |  \exp\left\{|1-\gamma|\left(|\xi^\eta_T|+\frac{1}{2} \int_0^T\eta^2_t\Sigma^2_t\,dt\right)\right\},
\]
which is integrable. Dominated convergence yields the desired result.
 	\end{proof}

 \begin{theorem}\label{thm:one-to-one-power}
		Let $\pi^*\in H^2_{BMO}$ be the best response of $\nu$. Then it must satisfy
		\[
						\pi^* = \frac{h+\sigma Z+\sigma^0Z^0}{\gamma\Sigma^2},
		\]
		where $(Z,Z^0)$ together with some $Y$ is the solution to the following BSDE 
		\begin{equation}\label{power:BSDE}
			\left\{\begin{split}
				-dY_t=&~\left\{	\frac{1}{2}(Z^2_t+(Z^0_t)^2) + \frac{ (1-\gamma)( h_t+\sigma_tZ_t+\sigma^0_t Z^0_t )^2 }{2\gamma\Sigma^2_t}		\right\}\,dt - Z_t\,dW_t - Z^0_t\,dW^0_t,\\
				Y_T=&~-(1-\gamma)\theta\widehat\nu_T.
			\end{split}\right. 
		\end{equation}
	Moreover, $(Z,Z^0)\in H^2_{BMO}\times H^2_{BMO}$.
 \end{theorem}
 \begin{proof}
 	Define $P_t=\mathbb [P_T|\mathcal F_t]$. By Assumption \ref{ass:power-investment}, there exists $(\lambda,\lambda^0)\in H^2_{BMO}\times H^2_{BMO}$ such that 
 	\begin{equation}\label{exponential-martingale:p}
 			P_t = P_0\mathcal E\left(    \int_0^t \lambda_s\,dW_s+\lambda^0_s\,dW^0_s        \right).
 	\end{equation}
 	Tower property of conditional expectation implies that 
 	\begin{equation}\label{power:int-by-parts-1}
 		\begin{split}
 			\mathbb E\left[  P_T\int_0^T\eta_t(	h_t - \pi^*_t\Sigma^2_t	)\,dt    \right] = 	\mathbb E\left[  \int_0^T\eta_t(	h_t - \pi^*_t\Sigma^2_t	)\mathbb E\left[ P_T|\mathcal F_t\right]\,dt    \right]=\mathbb E\left[  \int_0^T\eta_t(	h_t - \pi^*_t\Sigma^2_t	) P_t\,dt    \right].
 		\end{split}
 	\end{equation}
 	Integration by parts implies that by noting that $P$ is a true martingale 
 	\begin{equation}\label{power:int-by-parts-2}
 		\mathbb E\left[ P_T\left(\int_0^T \eta_t\sigma_t\,dW_t + \int_0^T \eta_t\sigma^0_t\,dW^0_t\right)   \right]=\mathbb E\left[  \int_0^T\eta_t(\sigma_tp_t\lambda_t+\sigma^0_tp_t\lambda^0_t) \,dt	   \right].
 	\end{equation}
 Summing up \eqref{power:int-by-parts-1} and \eqref{power:int-by-parts-2}, and noting the definition of $\xi^\eta_T$, we have 
 \begin{equation*}
 	\mathbb E\left[ P_T\xi^\eta_T  \right] = \mathbb E\left[ \int_0^T \eta_tP_t\left(  	h_t-\pi^*_t\Sigma^2_t+\sigma_t\lambda_t+\sigma^0_t\lambda^0_t	\right)\,dt     \right].
 \end{equation*}
By \eqref{variation:J} and the optimality of $\pi^*$, using the same argument as Proposition \ref{prop:maximum-principle}, we have 
\begin{equation}\label{power:maximum-principle}
		 h-\pi^*\Sigma^2+\sigma\lambda+\sigma^0\lambda^0=0.
\end{equation}
Define 
	\begin{equation}\label{power:YZ}
		\left\{\begin{split}
			Y=&~\log P-(1-\gamma)\widehat X^*,\\
			Z = &~\lambda - (1-\gamma)\sigma\pi^*,\\
			Z^0 = &~\lambda^0-(1-\gamma)\sigma^0\pi^*,
		\end{split}\right.
	\end{equation}
which implies that 
\begin{equation}\label{power:candidate-Y}
	dY_t=\left\{	-\frac{1}{2}(\lambda^2_t+(\lambda^0_t)^2) -(1-\gamma)\left( \pi^*_th_t-\frac{1}{2}(\pi^*_t)^2\Sigma^2_t		\right)				\right\}\,dt +Z_t\,dW_t+Z^0_t\,dW^0_t, 
\end{equation}
and that $Y_T=\log P_T-(1-\gamma)\widehat X^*_T=-(1-\gamma)\theta\widehat\nu_T$.
Taking \eqref{power:YZ} into \eqref{power:maximum-principle}, we have 
\[
		\pi^* = \frac{h+\sigma Z+\sigma^0Z^0}{\gamma\Sigma^2}.
\]
Taking the above equation into the driver of $Y$ in \eqref{power:candidate-Y}, we obtain \eqref{power:BSDE}. 
 \end{proof}
 
 \begin{remark}
We reiterate once again that the strong integrability condition for $X_T$ in Definition \ref{def:admissibility}, although standard in the literature, is specifically tailored to the Epstein–Zin utility framework. In the case of power utility, Theorem \ref{thm:one-to-one-power} yields the same necessary characterization as in \cite{Fu-2021} under exactly the same assumptions there. In particular, the integrability condition on $X_T$ imposed in Definition \ref{def:admissibility} is not required in this setting.

In the next section, we further demonstrate that, even in the presence of consumption under power utility, the space $\bigcap_{\beta\in\mathbb R}L^{\beta,T}_+$ remains unnecessary. Moreover, our approach improve the result in \cite{Fu2023} by enlarging the space in which the one-to-one correspondence holds.
	\end{remark}

	\subsection{Power utility with both investment and consumption}

	In this section, we revisit the problem in \cite{Fu2023}: for a fixed positive-valued process $\nu$,
	\[
			J(\pi,c) = \mathbb E\left[	  U(X_T^{\pi,c}\nu^{-\theta}_T) + a \int_0^T U(c_sX^{\pi,c}_s \nu_s^{-\theta})	 \,ds		\right]\rightarrow\max_{\pi,c},
	\]
	where the wealth process follows
	\[
			dX^{\pi,c}_t = \pi_t X^{\pi,c}_t\left( h_t\,dt + \sigma_t\,dW_t + \sigma^0_t\,dW^0_t   \right)-c_tX^{\pi,c}_t\,dt.
	\]
	and $a>0$ is a weight between running utility and terminal utility.
	
	Let $(\pi^*,c^*)\in H^2_{BMO}\times L^{\infty}_{\text{log},+}$ be the best response.	
	Since the sign of $1-\gamma$ is uncertain, we define the adjusted running and terminal utilities, thus positive, as
	\begin{equation}\label{adjusted-running-utility:power-consumption}
			\mathcal R^*_t= (c^*_t)^{1-\gamma} e^{(1-\gamma)\widehat X^*_t}\nu_t^{-(1-\gamma)\theta}
	\end{equation}
and 
	\begin{equation}\label{adjusted-terminal-utility:power-consumption}
	\mathcal T^*=e^{(1-\gamma)\widehat X^*_T}\nu^{-(1-\gamma)\theta}_T.
\end{equation}
	Define the optimal conditional expected utility as
	\[
			D^*_t = \mathbb E\left[	\left.	\frac{1}{1-\gamma} \mathcal T^* + a\int_0^T\frac{1}{1-\gamma} \mathcal R^*_s\,ds	\right|\mathcal F_t\right],
	\]
	 define
	\[
			N^*_t = \mathbb E\left[\left. \frac{1}{1-\gamma}\mathcal T^*+a\int_t^T\frac{1}{1-\gamma}\mathcal R^*_s\,ds \right| \mathcal F_t  \right]
	\]
	and define 
\begin{equation}\label{def:P-power-consumption}
		P=(1-\gamma)N^*.
	\end{equation}
    Martingale representation yields $(K,K^0)$ such that
    \[
        dP_t=-a\mathcal R^*_t\,dt + K_t\,dW_t+K^0_t\,dW^0_t=-a\mathcal R^*_t\,dt+(1-\gamma)\, dD^*_t.
    \]
    Define
    \begin{equation}\label{eq:lambda-power-consumption}
        (\lambda,\lambda^0)=\left( \frac{K}{P}, \frac{K^0}{P}   \right).
    \end{equation}
	Throughout this section, we make the following assumption. 
\begin{ass}\label{ass:power-investment-consumption}
    $		(\lambda,\lambda^0) \in H^2_{BMO}\times H^2_{BMO}
$.
\end{ass}
\begin{remark}
    Define
    \begin{equation}\label{eq:mathring-Z-power-consumption}
        (\mathring Z,\mathring Z^0)=\left( \frac{N^*}{D^*}\lambda, \frac{N^*}{D^*}\lambda^0     \right).
    \end{equation}
    Then we have
    \[
        dD^*_t = D^*_t\mathring Z_t\,dW_t+D^*_t\mathring Z^0_t\,dW^0_t,
    \]
    which implies that $		D^*_t = D^*_0\mathcal E\left(  \int_0^t \mathring Z_s\,dW_s + \int_0^t \mathring Z^0_s\,dW^0_s     \right)
$ and that
\begin{equation}\label{BSDE-P-1:power-consumption}
	dP_t=-a\mathcal R^*_t\,dt + (1-\gamma)D^*_t\mathring Z_t\,dW_t+(1-\gamma)D^*_t\mathring Z^0_t\,dW^0_t,\quad P_T=\mathcal T^*.
\end{equation}
Moreover, by the definition of $D^*$ and $N^*$, we have
\[
    |\mathring Z|+|\mathring Z^0|\leq |\lambda|+|\lambda^0|,
\]
which implies $(\mathring Z,\mathring Z^0)\in H^2_{BMO}\times H^2_{BMO}$. Thus,
\begin{equation}\label{ass-1:power-consumption}
		D^* \text{ satisfies }R_b \text{ for some }b>1.
	\end{equation}
    In particular, if there is no consumption, i.e., $a=0$, then $N^*=D^*$. Assumption \ref{ass:power-investment-consumption} reduces to \eqref{ass-1:power-consumption}, which coincides with Assumption \ref{ass:power-investment}.
\end{remark}


Let the best response satisfy	\begin{equation*}\label{ass-3:power-consumption}
		(\pi^*,\log c^*)\in H^2_{BMO}\times L^\infty.
	\end{equation*}
	For each $\epsilon>0$ and progressively measurable and bounded processes $(\eta,\kappa)$, define the local perturbation $(\pi^\epsilon,c^\epsilon)=(\pi^*+\epsilon\eta,c^*e^{\epsilon\kappa})$, the same as Proposition \ref{prop:maximum-principle}. The associated log-wealth process $\widehat X^\epsilon$ satisfies \eqref{eq:log-wealth-epsilon}.

	The following lemma establishes the first order variation of the expected utility.
	\begin{lemma}
 The first order variation of the log-wealth $\mathcal X= \lim_{\epsilon\rightarrow 0}  \frac{\widehat X^\epsilon-\widehat X^*}{\epsilon}$ satisfies
			\begin{equation}\label{variation-X:power-consumption}
		d\mathcal X_t = d\xi^\eta_t -c^*_t\kappa_t\,dt,\quad \mathcal X_0=0,			
	\end{equation}
	where $\xi^\eta$ satisfies \eqref{eq:xi-eta-t}.

		The first order variation of the expected utility satisfies 
		\begin{equation}\label{variation-J:power-consumption}
			0=\left.\frac{d}{d\epsilon} J(\pi^\epsilon,c^\epsilon)\right\vert_{\epsilon=0} = \mathbb E\left[ \mathcal T^*\mathcal X_T+a\int_0^T \mathcal R^*_t(\mathcal X_t+\kappa_t)\,dt   \right].
		\end{equation}
	\end{lemma}
	\begin{proof}
		From the dynamics of $\widehat X^\epsilon$ and $\widehat X^*$ in \eqref{eq:log-wealth-epsilon} and \eqref{eq:optimal-log-wealth}, it follows immediately that
			\begin{equation}\label{diff-log-weath:power-consumption}
			\widehat X^\epsilon_t  - \widehat X^*_t = \epsilon \xi^\eta_t-\int_0^t c^*_s(e^{\epsilon\kappa_s}-1)\,ds - \frac{\epsilon^2}{2}\int_0^t \eta^2_s\Sigma^2_s\,ds,
		\end{equation}
	which implies $\eqref{variation-X:power-consumption}$ by dividing $\epsilon$, taking limit and using dominated convergence.

		Using the equalities about the terminal utility
		\begin{equation}\label{diff-terminal:power-consumption}
			\begin{split}
				U(X^\epsilon_T\nu^{-\theta}_T) - U(X^*_T\nu^{-\theta}_T) =&~ \frac{1}{1-\gamma}\nu^{-\theta(1-\gamma)}_T\left( e^{\widehat X^\epsilon_T(1-\gamma)} - e^{\widehat X^*_T(1-\gamma)}  \right)\\
				=&~ e^{\widehat X^*_T(1-\gamma)}\nu^{-\theta(1-\gamma)} (	\widehat X^\epsilon_T-\widehat X^*_T	) \int_0^1 e^{(1-\gamma) \varrho (\widehat X^\epsilon_T-\widehat X^*_T)  }\,d\varrho \\
				=&~\mathcal T^*(	\widehat X^\epsilon_T-\widehat X^*_T	) \int_0^1 e^{(1-\gamma) \varrho (\widehat X^\epsilon_T-\widehat X^*_T)  }\,d\varrho
			\end{split}
		\end{equation}
	and about the running utility 
	\begin{equation}\label{diff-running:power-consumption}
		\begin{split}
			U( c^\epsilon X^\epsilon\nu^{-\theta} ) - U(c^* X^* \nu^{-\theta}) = &~
			\frac{1}{1-\gamma}\nu^{-\theta(1-\gamma)} \left(  (c^\epsilon)^{1-\gamma} e^{(1-\gamma)\widehat X^\epsilon}  - (c^*)^{1-\gamma}e^{(1-\gamma)\widehat X^*}  \right)\\
			=&~\frac{1}{1-\gamma}\nu^{-\theta(1-\gamma)}(c^*)^{1-\gamma} e^{(1-\gamma)\widehat X^*}\left(  e^{(1-\gamma)\epsilon\kappa + (1-\gamma)(\widehat X^\epsilon-\widehat X^*)}-1     \right)\\
			=&~\mathcal R^*\left(  \epsilon\kappa+\widehat X^\epsilon - \widehat X^*  \right) \int_0^1 e^{	 \varrho(1-\gamma)\epsilon\kappa + \varrho(1-\gamma)(\widehat X^\epsilon - \widehat X^*)		} \,d\varrho,
		\end{split}
	\end{equation}
we have 
		\begin{equation}\label{variation-J-eps:power-consumption}
			\begin{split}
				\frac{J(\pi^\epsilon,c^\epsilon) - J(\pi^*,c^*)}{\epsilon} =&~ \mathbb E\left[  \mathcal T^* \frac{	\widehat X^\epsilon_T-\widehat X^*_T}{\epsilon}	 \int_0^1 e^{(1-\gamma) \varrho (\widehat X^\epsilon_T-\widehat X^*_T)  }\,d\varrho  \right]\\
				&~+ \mathbb E\left[  a\int_0^T \mathcal R^*_t\left(   \kappa_t+\frac{\widehat X^\epsilon_t - \widehat X^*_t}{\epsilon}  \right) \int_0^1 e^{	 \varrho(1-\gamma)\epsilon\kappa_t + \varrho(1-\gamma)(\widehat X^\epsilon_t - \widehat X^*_t)		} \,d\varrho\,dt     \right].
			\end{split}
		\end{equation}
	To obtain \eqref{variation-J:power-consumption}, on the one hand, from \eqref{diff-terminal:power-consumption} and \eqref{diff-running:power-consumption}, we have the pointwise convergence
	\[
			\lim_{\epsilon\rightarrow\infty}\frac{U(X^\epsilon_T\nu^{-\theta}_T) - U(X^*_T\nu^{-\theta}_T) }{\epsilon} = \mathcal T^*\mathcal X_T
	\]
	and
	\[
			\lim_{\epsilon\rightarrow\infty}\frac{U( c^\epsilon X^\epsilon\nu^{-\theta} ) - U(c^* X^* \nu^{-\theta}) }{\epsilon} = \mathcal R^*(\kappa +\mathcal X).
	\]
	On the other hand, if we can bound the random variables inside the expectations in \eqref{variation-J-eps:power-consumption} by an integrable random variable, then dominated convergence yields \eqref{variation-J:power-consumption}. 
	
	It remains to establish the bound. By the boundedness of $\eta$, $\kappa$ and $\pi^*$, from \eqref{diff-log-weath:power-consumption} we have
	\[
		\left|\widehat X^\epsilon_t - \widehat X^*_t\right| +	\frac{\left|\widehat X^\epsilon_t - \widehat X^*_t\right|}{\epsilon} \leq \text{const}\left(1+\int_0^t|\pi^*_s|\,ds + \left|\int_0^t\eta_s\sigma_s\,dW_s\right|+\left|\int_0^t\eta_s\sigma^0_s\,dW^0_s\right|\right):=\mathcal B_t, 
	\]
	where the constant is independent of $t$ and $\epsilon$. Then, by noting $\mathcal T^*$ and $\mathcal R^*$ are positive, we have
	\begin{equation*}
		\begin{split}
			   &~\mathcal T^* \frac{	\widehat X^\epsilon_T-\widehat X^*_T}{\epsilon}	 \int_0^1 e^{(1-\gamma) \varrho (\widehat X^\epsilon_T-\widehat X^*_T)  }\,d\varrho +  a\int_0^T \mathcal R^*_t\left(   \kappa_t+\frac{\widehat X^\epsilon_t - \widehat X^*_t}{\epsilon}  \right) \int_0^1 e^{	 \varrho(1-\gamma)\epsilon\kappa_t + \varrho(1-\gamma)(\widehat X^\epsilon_t - \widehat X^*_t)		} \,d\varrho\,dt     \\
			   \leq&~\left( \mathcal T^* +a\int_0^T\mathcal R^*_t\,dt\right) \left(   \|\kappa\| + \sup_{0\leq t\leq T}\mathcal B_t 		    \right)  \left(1+ e^{|1-\gamma|\|\kappa\|}\right) e^{|1-\gamma|\sup_{0\leq t\leq T}\mathcal B_t}.
		\end{split}
	\end{equation*}
Since $\pi^*\in H^2_{BMO}$, $\sup_{0\leq t\leq T}\mathcal B_t$ is integrable to any power. Moreover, John-Nirenberg inequality yields a constant $\varepsilon>0$ small enough such that for any $\beta$, it holds
\begin{equation*}
	\mathbb E\left[   e^{\beta \int_0^T|\pi^*_s|\,ds}  \right]\leq \mathbb E\left[   e^{ \varepsilon \int_0^T|\pi^*_s|^2\,ds+\frac{\beta^2T}{4\varepsilon}}  \right]<\infty. 
\end{equation*}
By \eqref{ass-1:power-consumption}, we have $\mathcal T^* +a\int_0^T\mathcal R^*_t\,dt\in L^r$ for some $r>1$. Thus, the upper bound is integrable. 
	\end{proof}
The next lemma establishes the maximum principle.
\begin{lemma}
	Recall $P$ defined in \eqref{BSDE-P-1:power-consumption} and $(\lambda,\lambda^0)$ defined in \eqref{eq:lambda-power-consumption}.
	The best response satisfies
	\begin{equation}\label{optimality-1:power-consumption}
			h-\pi^*\Sigma^2+ \sigma\lambda+\sigma^0\lambda^0	=0
	\end{equation}
and
	\begin{equation}\label{optimality-2:power-consumption}
	 	a\mathcal R^* - c^* P =0.
\end{equation} 
\end{lemma}
\begin{proof}
	By \eqref{BSDE-P-1:power-consumption} and \eqref{variation-X:power-consumption}, integration by parts implies that 
	\[
		d(P_t\mathcal X_t) = \left\{  P_t\eta_t(h_t-\pi^*_t\Sigma^2_t) - P_tc^*_t\kappa_t -a\mathcal R^*_t\mathcal X_t \right\}\,dt +\left\{		(1-\gamma)D^*_t\mathring Z_t\eta_t\sigma_t+(1-\gamma)D^*_t\mathring Z^0_t\eta_t\sigma_t			\right\}\,dt +\text{local martingale}.
	\]
	By the definition of $P$ in \eqref{def:P-power-consumption} and the definition of $(\mathring Z,\mathring Z^0)$ in \eqref{eq:mathring-Z-power-consumption}, we have
	\[
		(1-\gamma)D^* \mathring Z\eta\sigma=P\lambda\sigma\eta,\quad  (1-\gamma)D^* \mathring Z^0\eta\sigma^0=P\lambda^0\sigma^0\eta,
	\]
	which implies
	\[
			d(P_t\mathcal X_t) = \left\{ P_t(h_t-\pi^*_t\Sigma^2_t)\eta_t - P_t c^*_t\kappa_t-a\mathcal R^*_t\mathcal X_t+P_t\lambda_t\sigma_t\eta_t+P_t\lambda^0_t\sigma^0_t\eta_t			\right\}\,dt + \text{local martingale}.
	\]
	Let $\tau_n\uparrow T$ be a sequence of localization stopping times. Then it holds that
	\begin{equation}\label{localization:power-consumption}
		\mathbb E\left[ P_{\tau_n}\mathcal X_{\tau_n}  \right] = \mathbb E\left[ \int_0^{\tau_n}\left\{ P_t(h_t-\pi^*_t\Sigma^2_t)\eta_t - P_t c^*_t\kappa_t-a\mathcal R^*_t\mathcal X_t+P_t\lambda_t\sigma_t\eta_t+P_t\lambda^0_t\sigma^0_t\eta_t			\right\}\,dt    \right].
	\end{equation}
	Since $\mathcal T^*>0$ and $\mathcal R^*_t>0$, we have
	\[
		0<P_t=\mathbb E\left[ \left. \mathcal T^* + a\int_t^T\mathcal R^*_s\,ds \right|\mathcal F_t  \right]\leq \mathbb E\left[ \left. \mathcal T^* + a\int_0^T\mathcal R^*_s\,ds \right|\mathcal F_t  \right],
	\]
	which together with \eqref{ass-1:power-consumption} and Doob's maximal inequality implies that for some $r>1$
	\[
			\mathbb E\left[ \sup_{0\leq t\leq T} P_t^r     \right]<\infty.
	\]
	Since, $\pi^*\in H^2_{BMO}$, by \eqref{variation-X:power-consumption}, $\sup_{0\leq t\leq T}|\mathcal X_t|$ is integrable to all powers. Thus,
	\[
			\lim_{n\rightarrow\infty}	\mathbb E\left[ P_{\tau_n}\mathcal X_{\tau_n}  \right] =\mathbb E\left[  P_T\mathcal X_T \right] = \mathbb E[\mathcal T^*\mathcal X_T] = -\mathbb E\left[ a \int_0^T\mathcal R^*_t(\mathcal X_t+\kappa_t)\,dt    \right],
	\]
	where the last equality is given by \eqref{variation-J:power-consumption}.
	
	Since $(\pi^*,c^*)\in  H^2_{BMO}\times L^\infty_{\text{log},+}$ and $(\lambda,\lambda^0)\in H^2_{BMO}\times H^2_{BMO}$ by Assumption \ref{ass:power-investment-consumption}, we can pass the limit inside the right hand side of \eqref{localization:power-consumption} to obtain
	\[
			\mathbb E\left[  \int_0^T P_t\eta_t\left(  h_t-\pi^*_t\Sigma^2_t+\lambda_t\sigma_t+\lambda^0_t\sigma^0_t   \right)  +\kappa(  a\mathcal R^*_t-P_tc^*_t) \,dt \right]=0,
	\]
	which yields \eqref{optimality-1:power-consumption} and \eqref{optimality-2:power-consumption} by the same argument as in Proposition \ref{prop:maximum-principle}. 
\end{proof}

 We are ready to conclude this section with the following theorem.
 
 \begin{theorem}\label{thm:power-consumption}
 The best response must satisfy
 \begin{equation}\label{necessary-condition:power-consumption}
 		c^* = \left(  a e^{-Y} \nu^{-(1-\gamma)\theta}   \right)^{\frac{1}{\gamma}},\quad \pi^* = \frac{h+\sigma Z+\sigma^0Z^0}{\gamma\Sigma^2},
 \end{equation}
  where $(Y,Z,Z^0)$ satisfies
 	\begin{equation}\label{BSDE:power-consumption}
 		\left\{\begin{split}
 		-dY_t =&~\left\{ \frac{1}{2}(Z^2_t+(Z^0_t)^2) + \frac{1-\gamma}{2\gamma} \frac{(h_t+\sigma_t Z_t+\sigma^0_tZ^0_t)^2}{\Sigma^2_t}+\gamma \left(  a e^{-Y_t} \nu^{-(1-\gamma)\theta}_t   \right)^{\frac{1}{\gamma}}\right\}\,dt\\
        &~-Z_t\,dW_t - Z^0_t\,dW^0_t,\\
 		Y_T=&~-(1-\gamma)\theta\log\nu_T,
 		\end{split}\right. 
 	\end{equation}
with $(Z,Z^0)\in H^2_{BMO}\times H^2_{BMO}$.
 \end{theorem}
\begin{proof}
	Define
	\begin{equation}\label{power-consumption:YZ}
		\begin{split}
			Y = \log P-(1-\gamma)\widehat X^*,\quad Z=\lambda - (1-\gamma)\sigma\pi^*,\quad Z^0=\lambda^0-(1-\gamma)\sigma^0\pi^*.
		\end{split}
	\end{equation}
We verify that $(Y,Z,Z^0)$ satisfies \eqref{BSDE:power-consumption}. The terminal condition satisfies immediately from \eqref{power-consumption:YZ} and \eqref{adjusted-terminal-utility:power-consumption}:
\[
		Y_T = \log P_T - (1-\gamma)\widehat X^*_T =\log\mathcal T^*-(1-\gamma)\widehat X^*_T=-(1-\gamma)\theta\log\nu_T.
\]
Taking $P=e^{(1-\gamma)\widehat X^*+Y}$ into \eqref{optimality-2:power-consumption} and noticing the definition of \eqref{adjusted-running-utility:power-consumption}, we have
\[
		c^* = \left(  a e^{-Y} \nu^{-(1-\gamma)\theta}   \right)^{\frac{1}{\gamma}}.
\]
Taking \eqref{power-consumption:YZ} into \eqref{optimality-1:power-consumption}, we have
\begin{equation}\label{candidate-pi:power-consumption}
		\pi^* = \frac{h+\sigma Z+\sigma^0Z^0}{\gamma\Sigma^2}.
\end{equation}
Taking \eqref{optimality-2:power-consumption} into the BSDE \eqref{BSDE-P-1:power-consumption}, we have
\[
		dP_t = -c^*_t P_t\,dt + P_t\lambda_t\,dW_t+ P_t\lambda^0_t \,dW^0_t,
\]
which implies by \eqref{power-consumption:YZ}
\begin{equation*}
	\begin{split}
		dY_t = \left\{  -\gamma c^*_t-\frac{1}{2}\left( \lambda^2_t+(\lambda^0_t)^2  \right) - (1-\gamma)\left(  \pi^*_t h_t-\frac{1}{2}(\pi^*_t)^2\Sigma^2_t  \right)     \right\}\,dt + Z_t\,dW_t + Z^0_t\,dW_t^0
	\end{split}
\end{equation*}
Using \eqref{power-consumption:YZ} and \eqref{candidate-pi:power-consumption}, we have 
\[
		\frac{1}{2}(\lambda^2+(\lambda^0)^2) + (1-\gamma)\pi^* h -\frac{1}{2}(1-\gamma)(\pi^*)^2\Sigma^2=\frac{1}{2}(Z^2+(Z^0)^2) + \frac{1-\gamma}{2\gamma} \frac{(h+\sigma Z+\sigma^0Z^0)^2}{\Sigma^2},
\] 
which yield \eqref{BSDE:power-consumption}.
\end{proof}

\begin{remark}
	If we assume $\nu \in L^\infty_{\text{log},+}$, which is the case in equilibrium (see e.g. Theorem \ref{thm:BSDE-to-NE} and Theorem \ref{thm:one-to-one}), then Theorem \ref{thm:power-consumption} implies that $Y\in L^\infty$. Moreover, the MOP (see e.g. Theorem \ref{thm:BSDE-to-NE}) immediately implies that any solution $(Y,Z,Z^0)\in L^\infty\times H^2_{BMO}\times H^2_{BMO}$ generates a best response $(\pi^*,c^*)\in H^2_{BMO}\times L^\infty_{\text{log},+}$.

 	Using DPP, it is proved in \cite{Fu2023} that in the random setting

(i) If $(Y,Z,Z^0)\in L^\infty\times L^\infty\times L^\infty$, then $(\pi^*,c^*)$ defined in \eqref{necessary-condition:power-consumption} is a best response with $(\pi^*,c^*)\in L^\infty\times L^\infty_{\text{log},+}$.

(ii) Conversely, if $(\pi^*,c^*)\in L^\infty\times L^\infty_{\text{log},+}$ is a best response, and $D^*$ satisfies $R_b$ for some $b>1$, then it must satisfy \eqref{necessary-condition:power-consumption} for some $(Y,Z,Z^0)$ solving \eqref{BSDE:power-consumption}. In general, the space for $(Z,Z^0)$ was not specified there, whereas in the deterministic case one has $(Z,Z^0)\in L^\infty\times L^\infty$.

Our theorem \ref{thm:power-consumption}, based on  maximum principle, strengthens this result in the random setting by establishing a genuine one-to-one correspondence. Namely, the space of best responses
\[
	\{ (\pi^*,c^*)\in H^2_{BMO}\times L^\infty_{\text{log},+}: (\pi^*,c^*)\text{ is a best response and}   \text{ Assumption }\ref{ass:power-investment-consumption} \text{ holds} \}
\]
is in bijection with the space of solutions to \eqref{BSDE:power-consumption} satisfying $(Z,Z^0)\in H^2_{BMO}\times H^2_{BMO}$.
\end{remark}
 
 
 \bibliography{Fu}

\end{document}